\documentclass[letterpaper, 12pt, journal, onecolumn]{IEEEtran}

\usepackage{amssymb,amsmath,color}
\usepackage{graphicx}
\usepackage{epsfig}
\usepackage{algorithm,algorithmic,xspace}
\usepackage{psfrag}
\usepackage{pstricks}
\usepackage{subfigure}

\usepackage{pgfplots}
\pgfplotsset{compat=1.3}

\newcommand{\diam}{\operatorname{diam}}

\newcommand{\subject}{\text{subject to}}
\newcommand{\maximize}{\text{maximize}}

\newcommand{\mc}{\mathcal}

\newcommand{\hi}{\mathit{h}}   

\newcommand{\Hc}{\mathcal{H}}

\newcommand{\Bc}{\mathcal{B}}

\newcommand{\Je}{J_{\underline{\epsilon}}}  
\newcommand{\Jeps}{J_{\epsilon}}

\newcommand{\Gc}{\mathcal{G}_{c}}

\newcommand{\outnbrs}{\mathcal{N}_O}
\newcommand{\innbrs}{\mathcal{N}_I}

\newcommand{\dist}{\operatorname{dist}}

\newcommand{\orc}{\operatorname{ORC} }

\newcommand{\CPC}{{Cutting-Plane~Consensus~}}

\newcommand{\1}{\boldsymbol{1}}

\newtheorem{theorem}{Theorem}[section]
\newtheorem{proposition}[theorem]{Proposition}
\newtheorem{corollary}[theorem]{Corollary}

\newtheorem{lemma}[theorem]{Lemma}

\newtheorem{applemma}{Lemma}[section]

\newtheorem{remark}[theorem]{Remark}

\newtheorem{assumption}[theorem]{Assumption}

\newcommand\oprocendsymbol{\hbox{$\square$}}
\newcommand\oprocend{\relax\ifmmode\else\unskip\hfill\fi\oprocendsymbol}
\def\eqoprocend{\tag*{$\square$}}

\graphicspath{{figs/}}
\begin{document}

\title{A Polyhedral Approximation Framework for\\ Convex and Robust Distributed
  Optimization}

\author{Mathias B{\"u}rger, Giuseppe Notarstefano, and Frank Allg{\"o}wer
  \thanks{M. B{\"u}rger and F. Allg{\"o}wer thank the German Research Foundation
    (DFG) for financial support within the Cluster of Excellence in Simulation
    Technology (EXC 310/1) at the University of Stuttgart.
     The research of G. Notarstefano has received funding from the European
    Community's Seventh Framework Programme (FP7/2007-2013) under grant
    agreement no. 224428 (CHAT).
    } 
  \thanks{ M. B{\"u}rger and F. Allg{\"o}wer are with the Institute for
    Systems Theory and Automatic Control, University of Stuttgart,
    Pfaffenwaldring 9, 70550 Stuttgart, Germany, \texttt{ \{mathias.buerger,
      frank.allgower\}@ist.uni-stuttgart.de}. }  \thanks{ Giuseppe Notarstefano
    is with the Department of Engineering, University of Lecce, Via per
    Monteroni, 73100 Lecce, Italy, \texttt{giuseppe.notarstefano@unile.it.}  } 
  }

\maketitle

\begin{abstract}
  In this paper we consider a general problem set-up for a wide class of convex
  and robust distributed optimization problems in peer-to-peer networks.
  In this set-up convex constraint sets are distributed to the network
  processors who have to compute the optimizer of a linear cost function subject
  to the constraints.
  We propose a novel fully distributed algorithm, named \emph{cutting-plane
    consensus}, to solve the problem, based on an outer polyhedral approximation
  of the constraint sets.
  Processors running the algorithm compute and exchange linear approximations of
  their locally feasible sets.
  Independently of the number of processors in the network, each processor
  stores only a small number of linear constraints, making the algorithm
  scalable to large networks.
  The cutting-plane consensus algorithm is presented and analyzed for the
  general framework.
  Specifically, we prove that all processors running the algorithm agree on an
  optimizer of the global problem, and that the algorithm is tolerant to node
  and link failures as long as network connectivity is preserved.
  Then, the cutting plane consensus algorithm is specified to three different
  classes of distributed optimization problems, namely (i) \emph{inequality
    constrained problems}, (ii) \emph{robust optimization problems}, and (iii)
  \emph{almost separable optimization problems with separable objective
    functions and coupling constraints}. For each one of these problem classes
  we solve a concrete problem that can be expressed in that framework and
  present computational results.  That is, we show how to solve: position
  estimation in wireless sensor networks, a distributed robust linear program
  and, a distributed microgrid control problem.
\end{abstract}

\section{Introduction}
The ability to solve optimization problems by local data exchange between
identical processors with small computation and communication capabilities is a
fundamental prerequisite for numerous decision and control systems.
Algorithms for such distributed systems have to work within the following
specifications \cite{Bullo2009}.
All processors running the algorithm are exactly identical and each processor
has only a small memory available. The data assigned by the algorithm to a
processor should be independent of the overall network size or only slowly
growing with the degree of the processor node in the network.
None of the processors has global information or can solve the problem
independently.


This paper addresses a class of optimization problems in distributed processor
networks with asynchronous communication.
Distributed, or peer-to-peer, optimization is related to parallel
\cite{Bertsekas1997} or large-scale optimization \cite{Lasdon2002}, but has to
meet further requirements, such as asynchronous communication and lack of shared
memory or coordination units.
Distributed optimization has gained significant attention in the last
years. Initially major attention was given to asynchronous distributed
subgradient methods \cite{Tsitsiklis1986}, \cite{AN-AO:09}.
Asynchronous distributed primal and dual subgradient algorithms are important
tools in network utility maximization and have been intensively studied from a
communication networks perspective, see \cite{Low1999}, \cite{Low2002}.

Combined with projection operations, subgradient methods
can also solve constrained optimization problems \cite{Nedic2010},
\cite{Low1999}.
In the last years, the research scope has been widened and now several different
distributed algorithms are explored, each suited for particular optimization
problems. Distributed Newton methods are proposed for Network Utility
Maximization \cite{Zargham2011}, \cite{Zargham2011a}, or unconstrained strongly
convex problems \cite{Zanella2011}.
Distributed variants of Alternating Direction Method of Multipliers (ADMM) have
been proposed for distributed estimation \cite{Schizas2008}, and in the wider
context of machine learning \cite{Boyd2010}.
The ADMM shows often a good convergence rate. However, one structural difference
between ADMM and distributed algorithms such as subgradient methods or the novel
method proposed in this paper has to be emphasized. In a centralized
implementation ADMM requires a coordination step. Distributed ADMM replaces this
central coordination with a consensus algorithm. However, this requires a
\emph{synchronization} between all processors in the network, i.e., all
processors have to switch synchronously between local computations and the
consensus algorithm. Fully distributed algorithms, as the one proposed here,
work \emph{asynchronously} and every processor can switch between local
computations and communications at its own pace.

An alternative research direction was established in \cite{GN-FB:06z}  and \cite{Notarstefano2009},
where distributed abstract optimization problems were considered.
A similar approach was explored in \cite{Burger2011a}, \cite{Burger2011b} for a distributed simplex algorithm that solves degenerate linear programs and multi-agent assignment problems. 
Some results on the use of distributed cutting-plane methods for robust optimization have been presented in \cite{Burger2012}. The results of \cite{Burger2012} are presented in the present paper in the wider context of general distributed optimization using cutting-plane methods.
%

\textbf{The contributions of this paper are as follows. }
Motivated by several important applications, we consider a general distributed
optimization framework in which each processor has knowledge of a convex
constraint set and a linear cost function has to be optimized over the
intersection of these constraint sets. It is worth noting that linearity of the
cost function is not a limitation and that strict convexity of the optimization
problem is not required.
A novel fully distributed algorithm named \emph{cutting plane consensus} is
proposed to solve this class of distributed problems.
The algorithm uses a polyhedral outer-approximation of the constraint set.
Processors performing the algorithm generate and exchange a \emph{small and
  fixed} number of linear constraints, which provide a polyhedral approximation
of the original optimization problem.
Then, each processor updates its local estimate of the globally optimal solution
as the minimal 2-norm solution of the approximate optimization problem.
We prove the correctness of the algorithm in the sense that all processor
asymptotically agree on a globally optimal solution.
We show that the proposed algorithm satisfies all requirements of peer-to-peer
processor networks. In particular, it requires only a strictly bounded local
memory and the communication is allowed to be asynchronous.
We also prove that the algorithm has an inherent tolerance against the failure
of single processors.

To highlight the generality of the proposed polyhedral approximation method, we
show how it can solve three different representations of the general distributed
convex program.
First, we consider constraint sets defined by nominal convex \emph{inequality
  constraints}.
Second, we discuss the method for a class of \emph{uncertain} or
\emph{semi-infinite constraints}. We show that the novel algorithm is capable of
computing robust solutions to uncertain problems in peer-to-peer networks.
Finally, we show that \emph{almost separable convex programs}, i.e., convex
optimization problems with separable objective functions and coupling
constraints, can be formulated in the general framework when their dual
representation is considered. Applied to this problem class the \CPC algorithm
can be seen as a fully distributed version of the classical Dantzig-Wolfe
decomposition, or column generation method, with no central coordinating master
program.

The general algorithm derived in the paper applies directly to each of the three
problem classes, and all convergence guarantees remain valid. We present for
each problem class a relevant decision problem, which can be solved by the novel
algorithm. In particular, it is shown that localization in sensor networks,
robust linear programming and distributed control of microgrids can be solved by the algorithm. Additionally, computational studies are presented which show that the
novel algorithm has an advantageous time complexity.


\textbf{Relation to other optimization methods:} The general problem formulation
of this paper is similar to the formulation considered in
\cite{Nedic2010}. However, while the approach in \cite{Nedic2010} requires a
projection operation, which might be computationally
  expensive
for some constraint classes, our approach requires only the knowledge of a
polyhedral approximation.
Additionally, our method works on general time-varying directed graphs, and does
not require a balanced communication.

The almost separable optimization problem setup studied in Section
\ref{sec.SeparableCost} is the classical setup for large scale
optimization.
Dual decomposition methods decompose these problems into a master program and
several subproblems.
Cutting-plane methods can be used to solve the master program, leading to
algorithms that originate in the classical Dantzig-Wolfe decomposition 
\cite{Dantzig1961}, \cite{Lasdon2002}.
The algorithm we propose differs significantly from classical decomposition
methods. Indeed, our algorithm performs in asynchronous peer-to-peer networks
with identical processors, without any central or coordinating master program.
 
A distributed ADMM implementation \cite{Boyd2010} uses an average consensus
algorithm to replace the update of the central master program.
While this allows to perform all computations decentralized, it requires a
\emph{synchronization} between the processors. Additionally, preforming a
consensus algorithm repeatedly might require many communication steps between
the processors.
In contrast, our method requires neither synchronized communication nor
a repeated averaging.

\textbf{The remainder of the paper is organized as follows.} The optimization
problem and the processor network model are introduced in Section
\ref{sec.ProblemFormulation}. In Section \ref{sec.OuterApproximation} the ideas
of polyhedral outer-approximation and minimal norm linear programming are
reviewed. The main contribution of this paper, the \CPC algorithm, is presented
in a general form in Section \ref{sec.Algorithm}, where also the correctness of
the algorithm and its fault-tolerance are proven.
The application of the algorithm to inequality constrained problems and to a
localization problem in sensor networks is presented in Section
\ref{sec.Inequality}.
In Section \ref{sec.Robustness} it is shown how the algorithm can be used to
solve distributed robust optimization problems, and a computational study is
presented, which compares the completion time of the novel algorithm to
an ADMM algorithm.
The application of the \CPC algorithm to almost separable convex optimization
problems and to distributed microgrid control is discussed in Section
\ref{sec.SeparableCost}. Finally, a concluding discussion is given in Section
\ref{sec.Conclusions}.

\section{Problem Formulation and Network Model} \label{sec.ProblemFormulation}

We consider a set of processors $V = \{1,\ldots,n\}$, each equipped with
communication and computation capabilities.  Each processor $i$ has knowledge of
a convex and closed constraint set $\mc{Z}_{i} \subset \mathbb{R}^{d}$.  The
processors have to agree on a decision vector $z \in \mathbb{R}^{d}$ maximizing
a linear objective over the intersection of all sets $\mc{Z}_{i}$. That is, the
processors have to solve the distributed convex optimization problem
\begin{align}
  \begin{split} \label{prob.Basic}
    \maximize &\quad c^{T}z \\
    \subject &\quad z \in \bigcap_{i=1}^{n} \mc{Z}_{i}.
  \end{split}
\end{align}
We denote the feasible set in the following as $\mc{Z} :=
\bigcap_{i=1}^{n}\mc{Z}_{i}$. We assume
that $\mc{Z}$ is non-empty and that \eqref{prob.Basic} has a finite optimal
solution.

The communication between the processors is modeled by a directed graph
(digraph) $\Gc=(V,E)$, named \emph{communication graph}.
The node set $V = \{1,\ldots,n\}$ is the set of processor identifiers, and the
edge set $E \subset \{1,\ldots,n\}^{2}$ characterizes the communication among
the processors.
If the edge-set does not change over time, the graph is called static otherwise
it is called time-varying.  We model the communication with time-varying
digraphs of the form $\Gc(t)=(V,E(t))$, where $t \in \mathbb{N}$ represents a
slotted universal time.
A graph $\Gc(t)$ models the communication in the sense that at time $t$ there is
an edge from node $i$ to node $j$ if and only if processor $i$ transmits
information to processor $j$ at time $t$.
The time-varying set of outgoing (incoming) \emph{neighbors} of node $i$ at time
$t$, i.e., the set of nodes to (from) which there are edges from (to) $i$ at
time $t$, is denoted by $\outnbrs(i,t)$ ($\innbrs(i,t)$).
In a static directed graph, the minimum number of edges between node $i$ and $j$
is called the \emph{distance} from $i$ to $j$ and is denoted by
$\dist(i,j)$. The maximum $\dist(i,j)$ taken over all pairs $(i,j)$ is the
\emph{diameter} of the graph $\Gc$ and is denoted by $\diam(\Gc)$.
A static digraph is said to be \emph{strongly connected} if for every pair of
nodes $(i,j)$ there exists a path of directed edges that goes from $i$ to $j$.
For the time-varying communication graph we rely on the concept of a jointly
strongly connected graph.

\begin{assumption}[Joint Strong Connectivity] \label{ass.PeriodicConnectivity}
  For every time instant $t \in \mathbb{N}$, the union digraph
  $\Gc^{\infty}(t):=\cup_{\tau = t}^{\infty} \Gc(\tau)$ is strongly
  connected. \oprocend
\end{assumption}

In this paper we develop a distributed, asynchronous algorithm solving problem
\eqref{prob.Basic} according to the network model described above.
Each processor stores a small set of data and transmits at each time instant
these data to its out-neighbors $\outnbrs(i,t)$.
It is worth noting that in general it is impossible to encode the convex set
$\mc{Z}_{i}$ with finite data. Thus, the information about the sets $\mc{Z}_{i}$
cannot be explicitly exchanged among the processors.

\section{Polyhedral Approximation and Minimal Norm Linear
  Programming} \label{sec.OuterApproximation}

We start recalling some important concepts form convex and linear optimization.
We will work in the following with half-spaces of the form $ \hi := \{ z: a^{T}z
- b \leq 0 \}, $
where $a \in \mathbb{R}^{d}$ and $b \in \mathbb{R}$.
A half-space is called a \emph{cutting-plane} if it satisfies the following
properties. Given a closed convex set $\mc{S} \subset \mathbb{R}^{d}$ and a
query point $z_{q} \notin \mc{S}$, a cutting-plane $h(z_{q})$ separates $z_{q}$
from $\mc{S}$, i.e., $a(z_{q}) \neq 0$ and
\begin{align} \label{eqn.CuttingPlaneBasic} a^{T}(z_{q})z \leq b(z_{q}) \quad
  \mbox{for\;all}\; z \in \mc{S}, \quad \mbox{and} \quad a^{T}(z_{q})z_{q} -
  b(z_{q}) = s(z_{q}) > 0.
\end{align}
The concept of cutting-planes leads to the first algorithmic primitive, the
cutting-plane oracle.
\begin{quote}
  \textbf{Cutting-Plane Oracle} $\orc(z_{q},\mc{S})$: queried at $z_{q} \in
  \mathbb{R}^{d}$ for the set $\mc{S}$. If (i) $z_{q} \notin \mc{S}$ then it
  returns a cutting-plane $\hi(z_{q})$, separating $z_{q}$ and $\mc{S}$,
  otherwise (ii) it asserts that $z_{q} \in \mc{S}$ and returns an empty $\hi$.
 \end{quote}

We make the following assumption on the cutting plane oracle, following the general cutting-plane framework of \cite{Eaves1971}.
\begin{assumption}\label{ass.Separator} The cutting-plane oracle $\orc(z_{q},\mc{S})$ is such that 
  (i) $\|a(z_{q})\|_{2} < \infty$ and (ii) $z_{q}(t) \rightarrow \bar{z}$ and
  $s(z_{q}(t)) \rightarrow 0$ implies that $\bar{z} \in \mc{S}.$
\end{assumption}
Note that this assumption is not very restrictive and holds for many important
problem formulations. In fact, we discuss three important problem classes for which the
assumption holds.

Given a collection of cutting-planes $H = \cup_{k=1}^{m} \hi_{k}$, the
polyhedron induced by these cutting-planes is $\mc{H} = \{ z : A_{H}^{T}z \leq
b_{H}\}$, with the matrix $A_{H} \in \mathbb{R}^{d \times m}$ as $A_{H}
=[a_{1},\ldots,a_{m}]$, and the vector $b_{H} = [b_{1},\ldots,b_{m}]^{T}$.

\begin{remark}[Cutting plane notation] We refer to both a half-space $h$ and the data
  inducing the half space with a small italic letter. A collection of
  cutting-planes is denoted with italic capital letters, e.g., $H =
  \bigcup_{k=1}^{m} h_{k}$. For a collection of cutting-planes, we denote the
  induced polyhedron with capital calligraphic letters, e.g., $\mc{H}$.
  Please note the following notational aspect. A collection of cutting-planes
  $B$ that is a subset of the cutting-planes contained in $H$ is denoted as $B
  \subset H$, while the induced polyhedra satisfy $\mc{B} \supseteq
  \mc{H}$. \oprocend
\end{remark}

Assume that each cutting-plane $h_{i}$ is generated as a separating hyperplane
for some set $\mc{Z}_{j}$, and let $H$ be a collection of cutting-planes. The
linear \emph{approximate program}
\begin{align}
  \begin{split} \label{prob.BasicApproxLP} \max_{z} &\quad c^{T}z \quad
    \mathrm{s.t.} \; A_{H}^{T}z \leq b_{H}
  \end{split}
\end{align}
is then a relaxation of the original optimization problem \eqref{prob.Basic}
since the polyhedron $\mc{H} = \{z:A_{H}^{T}z \leq b_{H}\}$ is an outer
approximation of the original constraint set $\mc{Z} =
\bigcap_{i=1}^{n}\mc{Z}_{i}$.
We denote in the following the optimal value of \eqref{prob.BasicApproxLP} as $
\gamma_{H} $, i.e., $\gamma_{H} := \max_{z \in \mc{H}} c^{T}z$. The linear
program \eqref{prob.BasicApproxLP} has in general several optimizers, and we
denote the set of all optimizers of \eqref{prob.BasicApproxLP} with
\begin{align}
  \Gamma_{H} := \{z \in \mc{H} : \; c^{T}z \geq c^{T}v, \forall v \in \mc{H}
  \}. \label{eqn.Gamma}
\end{align}
It is a standard result in linear programming that $\Gamma_{H}$ is always a
polyhedral set.
We consider throughout the paper the unique optimal solution to
\eqref{prob.BasicApproxLP} which has the minimal 2-norm, i.e., we aim to compute
\begin{align}
  z_{H}^{*} = \mathrm{arg} \min_{z \in \Gamma_{H}}
  \|z\|_{2}. \label{prob.LPprojection}
\end{align}
Finding a minimal norm solution to a linear program is a classical problem and
various solution methods are proposed in the literature. Starting from the early
reference \cite{Mangasarian1983} research on this topic is still actively
pursued \cite{Zhao2002}. The minimal 2-norm solution can be
efficiently computed as the solution of a quadratic program.
\begin{proposition} \label{prop.2normComputation} Let $u^{*} \in
  \mathbb{R}^{|H|},\alpha^{*} \in \mathbb{R}, l^{*}\in \mathbb{R}^{d}$ be the
  optimal solution to
  \begin{align}
    \begin{split} \label{prob.2NormDual}
      \min_{u, \alpha, l} \quad & \frac{1}{2} (A_{H}u + \alpha c)^{T}(A_{H}u + \alpha c) + b_{H}^{T}u + c^{T}l \\
      \mathrm{s.t.} \quad & A_{H}^{T}l - \alpha b \geq 0, \; u \geq 0
    \end{split}
  \end{align}
  then $ z^{*}_{H} = -A_{H}u^{*} - \alpha^{*}c $ solves
  \eqref{prob.LPprojection}. \oprocend
\end{proposition}
The proof of this result is presented in Appendix A.
The minimal 2-norm solution has the important property that it always maximizes
a strongly concave cost function.
\begin{lemma} \label{prop.QuadraticObjective} Let a set of cutting-planes define
  the polyhedron $\mc{H}$ and let $z^{*}_{H}$ be the minimal 2-norm solution to
  \eqref{prob.BasicApproxLP}. Consider the quadratically perturbed linear
  objective
$$J_{\epsilon}(z) = c^{T}z - \frac{\epsilon}{2}\|z\|^{2}_{2} $$
parametrized with a constant $\epsilon>0$. Then there exists a $\bar{\epsilon} >
0$ such that for any $\epsilon \in [0,\bar{\epsilon}]$
\begin{align}
  z^{*}_{H} = \mathrm{arg} \max_{z \in \mc{H}} \;
  J_{\epsilon}(z). \label{prob.QuadraticPerturbation}
\end{align} \oprocend
\end{lemma}
The proof of this result is very similar to the classical proof presented in
\cite{Mangasarian1979}. However, since the considered set-up is slightly
different and the result is fundamental for the methodologies developed in the
paper, we present the proof in Appendix B.

Any solution to a (feasible) linear program of the form
\eqref{prob.BasicApproxLP} is fully determined by at most $d$ constraints.
This is naturally also true for the minimal 2-norm solution of a linear program.
We formalize this property with the notion of \emph{basis}. Given a collection of
cutting-planes $H$, we say that a subset $B \subseteq H$ is a basis of $H$ if
the minimal 2-norm solution to the linear program defined with the constraint
set $B$, say $z_{B}^{*}$, is identical to the minimal 2-norm solution of the
linear program defined with the constraint set $H$, say $z_{H}^{*}$, i.e.,
$z_{B}^{*} = z_{H}^{*}$, while for any strict subset of cutting-planes $B'
\subset B$, it holds that $z_{B'}^{*} \neq z_{B}^{*}$.  For a feasible problem,
the cardinality of a basis is bounded by the dimension of the problem, i.e.,
$|B| \leq d$.
Throughout this paper, a basis is always considered to be a basis with respect
to the 2-norm solution of the linear program and a
basis computation requires to compute the solution to problem
\eqref{prob.2NormDual}. Note, however, that the active constraints at an optimal
point $z_{H}^{*}$ are always a superset of a basis at this point and are exactly
a basis if the problem is not degenerate.
Therefore, in most cases it will be sufficient to find the active constraints,
which are easy to detect.

\section{The Cutting-Plane Consensus Algorithm} \label{sec.Algorithm}

For a network of processors, we propose and analyze the \CPC algorithm to solve
distributed convex optimization problems of the form \eqref{prob.Basic}.

\subsection*{The Distributed \CPC Algorithm}
The algorithm to solve general distributed optimization problems
\eqref{prob.Basic} is as follows. \\

\begin{quote}
  \textbf{Cutting-Plane Consensus:}
  Processors store and update collections of cutting-planes. The cutting-planes
  stored by agent $i$ at iteration $t$ are always a basis of a corresponding
  linear approximate program \eqref{prob.BasicApproxLP}, and are denoted by
  $B^{[i]}(t)$.
  A processor initializes its local collection of cutting-planes $B^{[i]}_{0}$
  with a set of cutting-planes chosen such that $\mc{B}^{[i]}_{0} \supset
  \mc{Z}_{i}$ and $\max_{z \in \mc{B}^{[i]}_{0}} c^{T}z < \infty$.

  Each processor repeats then the following steps:
  \begin{enumerate}
  \item[(S1)] it transmits its current basis $B^{[i]}(t)$ to all its
    out-neighbors $\outnbrs(i,t)$ and receives the basis of its in-neighbors
    $Y^{[i]}(t) = \bigcup_{j \in \innbrs(i,t)} B^{[j]}(t)$;
  \item[(S2)] it defines $H_{tmp}^{[i]}(t) = B^{[i]}(t) \cup Y^{[i]}(t)$, and
    computes (i) a query point $z^{[i]}(t)$ as the minimal 2-norm solution to
    the approximate program \eqref{prob.BasicApproxLP}, i.e.,
$$ z^{[i]}(t) = \mathrm{arg} \min_{z \in \Gamma_{H_{tmp}^{[i]}(t)}} \|z\|_{2} $$ 
and (ii) a minimal set of active constraints $B_{tmp}^{[i]}(t)$;
\item[(S3)] it calls the cutting-plane oracle for the constraint set
  $\mc{Z}_{i}$ at the query point $z^{[i]}(t)$,
$$ h(z^{[i]}(t)) = \orc(z^{[i]}(t),\mc{Z}_{i});$$ 
\item[(S4)] it updates its collection of cutting planes as follows: if
  $z^{[i]}(t) \in \mc{Z}_{i}$ then $B^{[i]}(t+1) = B_{tmp}^{[i]}(t)$, otherwise
  $B^{[i]}(t+1)$ is set to the minimal basis of $B_{tmp}^{[i]}(t) \cup
  h(z^{[i]}(t)) $. \\
\end{enumerate}
\end{quote}

The four steps of the algorithm can be summarized as communication (S1),
computation of the query point (S2), generation of cutting-plane (S3) and
dropping of all inactive constraints (S4).
The \CPC algorithm is explicitly designed for the use in processor networks. We
want to emphasize here the following four aspects of the algorithm. \\
\begin{LaTeXdescription}
\item[Distributed Initialization:] Each processor can initialize the local
  constraint sets as a basis of the artificial constraint set $\{z \in
  \mathbb{R}^{d} : -M\1 \leq z \leq M\1 \}$ for some $M \gg 1$. If $M \in
  \mathbb{R}_{> 0}$ is chosen sufficiently large, the artificial constraints
  will be dropped during the evolution of the algorithm. \\
\item[Bounded Communication:] Each processor stores and transmits at most
  $(d+1)d$ numbers at a time. In particular, processors exchange bases of
  \eqref{prob.BasicApproxLP}, which are defined by not more than $d$
  cutting-planes. Each cutting-plane is fully defined by $d+1$ numbers. \\
\item[Bounded Local Computations:] Each processor has to compute locally the
  2-norm solution to a linear program with $d(|\innbrs(i,t)|+1)$ constraints. \\
\item[Asynchronous Communication:] The \CPC algorithm does not require a
  time-synchronization.  Each processor can perform its local computations at
  any speed and update its local state whenever it receives data from some of
  its in-neighbors. \\
\end{LaTeXdescription}

Due to these properties, the \CPC algorithm is particularly well suited for
optimization in large networks of identical processors.

\subsection*{Technical Analysis of the \CPC Algorithm} \label{sec.Analysis}
Before starting the proof of the algorithm correctness, we point out three important
technical properties related to its evolution:  
\begin{itemize}
\item The linear constraints stored by a processor form always a
  \emph{polyhedral outer-approximation} of the globally feasible set $\mc{Z}$. 
\item The cost-function of each processor is monotonically non-increasing over
  the evolution of the algorithm. 
\item If the communication graph $\Gc$ is a strongly connected \emph{static}
  graph, then after $\diam(\Gc)$ communication rounds, all processors in the
  network compute a query-point with a cost not worse than the best processor at
  the initial iteration. 
\end{itemize}
These properties provide an intuition about the functionality of the algorithm
and the line we will follow to prove its correctness.
They are formalized and proven rigorously in Lemma \ref{prop.TechIssues} in
Appendix~\ref{sec.AppendixProofs}. 
We are ready to establish the correctness of the \CPC algorithm. We start by
formalizing two auxiliary results which are also interesting on their own.
The first result states the convergence of the query points to the locally
feasible sets.
\begin{lemma}[Convergence] \label{prop.Convergence} Assume Assumption \ref{ass.Separator} holds.  Let $z^{[i]}(t)$ be the
  query point generated by processor $i$ performing the \CPC algorithm. Then,
  the sequence $\{z^{[i]}(t)\}_{t \geq 0}$ has a limit point in the set
  $\mc{Z}_{i}$, i.e., there exists $\bar{z}^{[i]} \in \mc{Z}_{i}$ such that
  \begin{align*}
    \lim_{t\rightarrow \infty} \; \| z^{[i]}(t) - \bar{z}^{[i]}\|_{2}
    \rightarrow 0. \eqoprocend
  \end{align*}
\end{lemma}

The second result shows that all processors in the network will reach an
agreement.

\begin{lemma}[Agreement] \label{prop.Agreement} Assume the communication network
  $\Gc(t)$ is jointly strongly connected. Let $z^{[i]}(t)$ be query points
  generated by the \CPC algorithm, then
  \begin{align*}
    \lim_{t \rightarrow \infty} \; \|z^{[i]}(t) -z^{[j]}(t)\|_{2} \rightarrow 0,
    \quad \mbox{for\; all\;} i,j \in \{1,\ldots,n\}.  \eqoprocend
  \end{align*}
\end{lemma}

The correctness of the algorithm is summarized in the following theorem.
\begin{theorem}[Correctness] \label{thm.AsymptoticConvergence} Let $\Gc(t)$ be a
  jointly strongly connected communication network with processors performing
  the \CPC algorithm, and let Assumption \ref{ass.Separator} hold. Let $z^{*}$ be the unique optimizer to
  \eqref{prob.Basic} with minimal 2-norm, then
  \begin{align*}
    \lim_{t \rightarrow \infty} \|z^{[i]}(t) - z^{*}\|_{2} \rightarrow 0 \quad
    \mbox{for\;all}\; i \in \{1,\ldots,n\}. \eqoprocend
  \end{align*}
\end{theorem}

For the clarity of presentation, the technical proofs of Lemma
\ref{prop.Convergence}, Lemma \ref{prop.Agreement}, and Theorem
\ref{thm.AsymptoticConvergence} are presented in Appendix
\ref{sec.AppendixProofs}.

A major advantage for using the \CPC algorithm in distributed systems is its
inherent fault-tolerance. The requirements on the communication network are very
weak and the algorithm can well handle disturbances in the communication like,
e.g., packet-losses or delays.
Additionally, the algorithm has an inherent tolerance against processor
failures. We say that a processor fails if it stops at some time $t_{f}$ to
communicate with other processors.

\begin{theorem}[Fault-Tolerance]
  Suppose that processor $l$ fails at time $t_{f}$, and that the communication
  network remains jointly strongly connected after the failure of processor
  $l$. Let $z^{[l]}(t_{f})$ be the last query point computed by processor $l$
  and define $\gamma^{[l]}(t_{f})=c^{T}z^{[l]}(t_{f})$.  Then the query-points
  computed by all processors converge, i.e., $ \lim_{t \rightarrow \infty}
  \|z^{[i]}(t) - \bar{z}_{-l}\| \rightarrow 0,$ with $\bar{z}_{-l}$ satisfying
$$\bar{z}_{-l} \in \left( \bigcap_{i \neq l} \mc{Z}_{i} \right) \quad 
\mbox{and} \quad c^{T}\bar{z}_{-l} \leq \gamma^{[l]}(t_{f}). $$
\end{theorem}

\begin{proof}
  Consider the evolution of the algorithm starting at time $t_{f}$.
  With Lemma \ref{prop.Convergence} and Lemma \ref{prop.Agreement} one can
  conclude that for all processors $i \neq l$, the query points will converge to
  the set $\left( \bigcap_{i \neq l} \mc{Z}_{i} \right)$.
  Additionally, the out-neighbors of the failing processor $l$ have received a
  basis $B^{[l]}(t_{f})$ such that the optimal value of the linear approximate
  program \eqref{prob.BasicApproxLP} is $\gamma^{[l]}(t_{f})$. Any query point
  $z^{[i]}(t), t \geq t_{f}$, subsequently computed by the out-neighbors of
  processor $l$ as the solution to \eqref{prob.BasicApproxLP} must therefore be
  such that $c^{T}z^{[i]}(t) \leq \gamma^{[l]}(t_{f})$ for all $t \geq t_{f}$.
\end{proof}
This last result provides directly a paradigm for the design of fault-tolerant
systems.
\begin{corollary}
  Suppose that for all $l \in V$, $ \bigcap_{i=1,i \neq l}^{n} \mc{Z}_{i} =
  \mc{Z}$. Then for all $l \in V$, $\bar{z}_{-l} = z^{*}$ with $z^{*}$ the
  optimal solution to \eqref{prob.Basic}. \oprocend
\end{corollary}

The abstract problem formulation \eqref{prob.Basic} and the \CPC algorithm
provide a \emph{general framework for distributed convex optimization}.
We show in the following that a variety of important representations of the
constraint sets are covered by this set-up.
Depending on the formulation of the local constraint sets $\mc{Z}_{i}$ different
cutting-plane oracles must be defined, leading to different realizations of the
algorithm.
We specify in the following the \CPC algorithm to three important problem
classes.
We want to stress that the correctness proofs established here for the general set-up
will hold directly for the three specific problem formulations discussed in the
remainder of the paper.

\section{Convex Optimization with Distributed Inequality
  Constraints} \label{sec.Inequality}

As first concrete setup, we consider the most natural realization of the general problem
formulation \eqref{prob.Basic} with the local constraint set defined by a
convex inequality, i.e.,
\begin{align} \label{prob.Inequality} \mc{Z}_{i} = \{z : f_{i}(z) \leq 0\}.
\end{align}
The functions $f_{i}: \mathbb{R}^{d} \mapsto \mathbb{R}$ are assumed to be
convex, but not necessarily differentiable.
Thus, the set-up \eqref{prob.Inequality} includes also the case in which
processor $i$ is assigned more that one constraint, say $\mc{Z}_{i} = \{z :
f_{i1}(z) \leq 0, f_{i2}(z) \leq 0, \ldots, f_{ik}(z) \leq 0 \}$. In fact, one
can define $f_{i}(z) := \max_{j \in \{1,\ldots k\}} f_{ij}(z)$ and directly
obtain the formulation \eqref{prob.Inequality}.

To define a cutting-plane oracle for constraints of the form
\eqref{prob.Inequality}, we use the concept of subdifferential. Given a
query-point $z_{q} \in \mathbb{R}^{d}$, the subdifferential of $f_{i}$ at
$z_{q}$ is
\begin{align*}
  \partial f_{i}(z_{q}) = \{g_{i} \in \mathbb{R}^{d} : f_{i}(z) - f_{i}(z_{q})
  \geq g_{i}^{T}(z - z_{q}), \; \forall z \in \mathbb{R}^{d} \}.
\end{align*}
An element $g_{i} \in \partial f_{i}(z_{q})$ is called a subgradient of $f_{i}$
at $z_{q}$. If the function $f_{i}$ is differentiable, then its gradient $\nabla
f_{i}(z_{q})$ is a subgradient.
A cutting-plane oracle for constraints of the form \eqref{prob.Inequality} is
now as follows, see, e.g., \cite{Kelley1960}.
\begin{quote} \textbf{Cutting-plane Oracle:} If a query point $z_{q}$ is such
  that $f_{i}(z_q) > 0$, then
  \begin{align}
    f_{i}(z_{q}) + g_{i}^{T}(z - z_{q}) \leq
    0, \label{eqn.CuttingPlaneSubgradient}
  \end{align}
  for some $g_{i} \in \partial f_{i}(z)$, is returned, .
\end{quote}
Note also that Assumption \ref{ass.Separator} is satisfied, since $s(z_{q}) =
f_{i}(z_{q}) + g_{i}^{T}(z_{q} - z_{q}) = f(z_{q})$, and $f(z_{q}) = 0$ implies
$z_{q} \in \mc{Z}_{i}$.
If $f_{i}(z) := \max_{j \in \{1,\ldots k\}} f_{ij}(z)$, then $\partial
f_{i}(z_{q}) = \mathbf{Co} \cup \{\partial f_{ij}(z_{q}) : f_{ij}(z_{q}) =
f_{i}(z_{q}) \}$, where $\mathbf{Co}$ denotes the convex hull. Thus, the method
is applicable for constraints where subgradients can be obtained.

\begin{remark}
  An important class of constraints are \emph{semi-definite} constraints of the
  form $ \mc{Z}_{i} = \{ z : F_{i}(z):=F_{i0} + z_{1}F_{i1} + \cdots +
  z_{d}F_{id} \leq 0 \}, $ where $F_{ij}$ are real symmetric matrices, and
  $'\leq 0'$ denotes negative semi-definite.
  The semi-definite constraint can be formulated as inequality constraint
  \begin{align*}
    f_{i}(z) := \lambda_{\max}(F_{i}(z)) \leq 0,
  \end{align*}
  with $\lambda_{\max}$ the largest eigenvalue of $F(z)$.
  It is discussed, e.g., in \cite{Scherer}, that given a query point $z_{q}$ and
  a normalized eigenvector $v_{q}^{*}$ of $F_{i}(z_{q})$ corresponding to
  $\lambda_{\max}(F_{i}(z_{q}))$, then the vector $g_{i} =
  [v_{q}^{*T}F_{1}v_{q}^{*}, \ldots, v_{q}^{*T}F_{d}v_{q}^{*}]^{T}$ is a
  subgradient of $f_{i}(z)$.
  The \CPC algorithm can thus handle semi-definite constraints and has to be
  seen in the context of the recent work on cutting-plane methods for
  semi-definite programming \cite{Krishnan2006}, \cite{Konno2003}. \oprocend
\end{remark}

The \CPC algorithm is directly applicable to problems where processors are
assigned convex, possibly non-differentiable, inequality constraints. Such
distributed problems appear in various important application. For example, the
distributed position estimation problem in wireless sensor networks can be
formulated in the form \eqref{prob.Basic} with convex inequality and
semi-definite constraints available only locally to (some of) the sensor nodes.

\subsection*{Application Example: Convex Position Estimation in Wireless Sensor
  Networks}

Wide-area networks of cheap sensors with wireless communication are envisioned
to be key elements of modern infrastructure systems. In most applications, only
few sensors are equipped with localization tools, and it is necessary to
estimate the position of the other sensors, see \cite{Bachrach2005}.

In \cite{Doherty2001} the sensor localization problem is formulated as a convex
optimization problem, which is then solved by a central unit using semidefinite
programming.
The semi-definite formulation proposed in \cite{Doherty2001} has been later
extended in the literature.
We formulate the distributed position estimation problem given in
\cite{Doherty2001} in the general distributed convex optimization framework
\eqref{prob.Basic} and show that the general \CPC algorithm can be used for a
fully distributed solution, using only local message passing between the
sensors.

Let in the following $\mathbf{v}_{i} \in \mathbb{R}^{2}$ denote the known
position of sensor $i \in \{1,\ldots,n\}$. We want to estimate the unknown
position of an additional sensor $z \in \mathbb{R}^{2}$.
In \cite{Doherty2001}, two different estimation mechanisms are considered: (i)
laser transmitters at nodes which scan through some angle, leading to a cone
set, which can be expressed by three linear constraints of the form $f(z) :=
a_{i}^{T}z -b_{i} \leq 0,$ $a_{i} \in \mathbb{R}^{2\times 1}$ and $b_{i} \in
\mathbb{R}$, two bounding the angle and one bounding the distance and (ii) the
range of the RF transmitter, leading to circular constraints of the form $ \|z -
\mathbf{v}_{i}\|_{2}^{2} \leq r_{i}^{2}$. Using the Schur-complement, the
quadratic constraint can be formulated as a semi-definite constraint of the form
$$ F_{i}(z) := (-1)\begin{bmatrix}r_{i}I_{2} & (z - \mathbf{v}_{i}) \\ (z-\mathbf{v}_{i})^{T} & r_{i} \end{bmatrix} \leq 0,$$
where $I_ {2}$ is the $2 \times 2$ identity matrix.
\begin{figure}[h!]
  \begin{center}
     \includegraphics[width = 0.45\textwidth, trim=4cm 13cm 5cm 4cm, ]{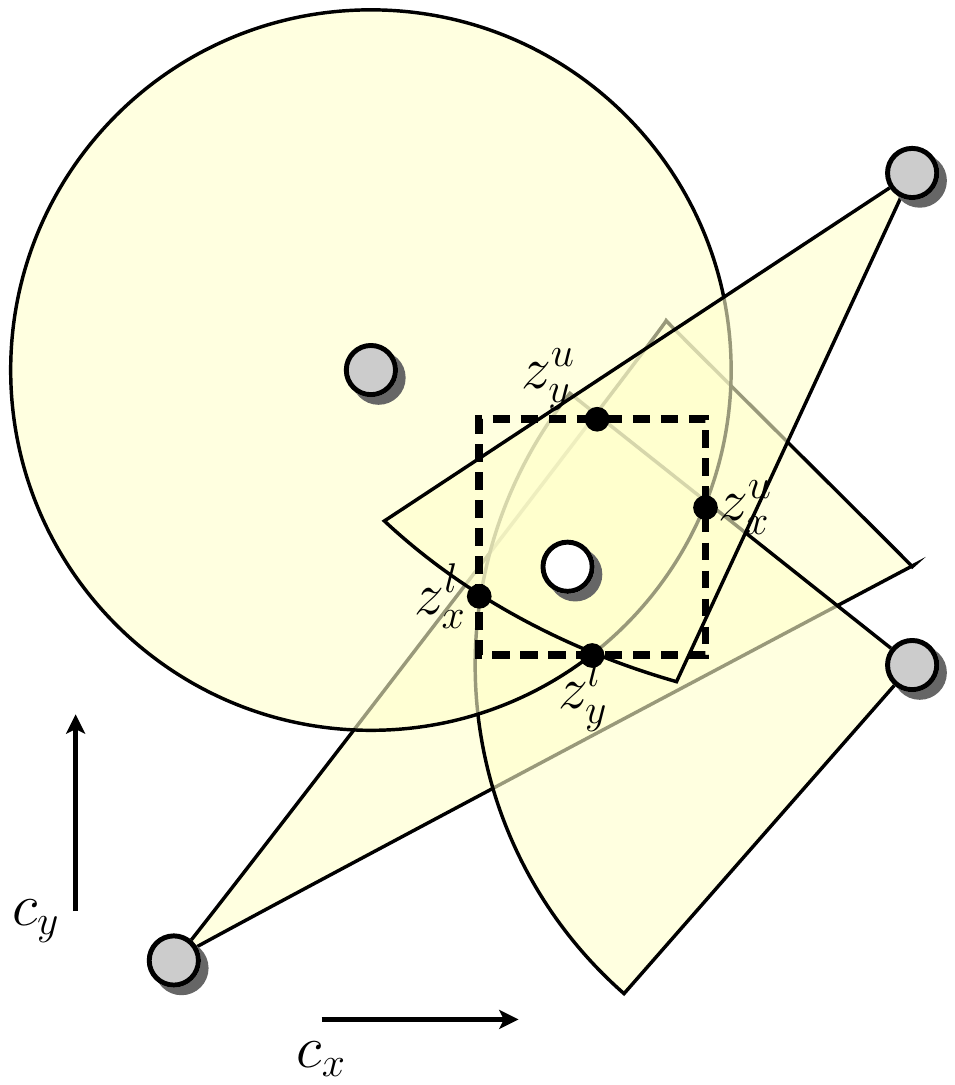} 
  \end{center}
  \caption{Localization of the white node by set estimates of the four gray
    nodes. The set estimate is given by the bounding box which determined by the
    four point $z_{x}^{l}, z_{x}^{u}, z_{y}^{l}, z_{y}^{u}$. The four extreme
    points can be found with the \CPC algorithm.} \label{fig.localization}
\end{figure}
Each sensor $i$ can bound the position of the unknown sensor to be contained in
the convex set $\mc{Z}_{i}$, which is, depending on the available sensing
mechanism, a disk represented by a semi-definite constraint $\mc{Z}_{i} = \{z :
F_{i}(z) \leq 0\}$, a cone $\mc{Z}_{i} = \{z : f_{ij}(z) \leq 0, j=1,2,3 \}$,
or a quadrant, $\mc{Z}_{i} = \{z : F_{ij}(z) \leq 0, f_{ij} \leq 0, j=1,2,3
\}$.

The sensing nodes can now compute the smallest box $\{ z \in \mathbb{R}^{2} :
[z^{l}_{x}, z^{l}_{x}]^{T} \leq z \leq [z^{u}_{x}, z^{u}_{y}]^{T}\}$ that is
guaranteed to contain the unknown position using the \CPC algorithm. As proposed
in \cite{Doherty2001}, the minimal bounding box can be computed by solving four
optimization problems with linear objectives.
To compute, for example, $z^{u}_{x}$ one defines the objective $c_{x} = [1,
0]^{T}$ and solves $ z^{u}_{x} := \max\; c_{x}^{T}z, \; \mbox{s.t.} \; z \in
\bigcap_{i=1}^{n}\mathcal{Z}_{i}. $ In the same way $z^{l}_{x}, z^{l}_{y},
z^{u}_{y}$ can be determined. Figure \ref{fig.localization} illustrates a
configuration where four nodes estimate the position of one node. A linear
version of such a distributed estimation problem, i.e., with all constraints
being linear inequalities, has been considered in the previous work
\cite{Notarstefano2009}.

\section{Robust Optimization with Uncertain Constraints} \label{sec.Robustness}

The general formulation \eqref{prob.Basic} covers also distributed robust
optimization problems with uncertain constraints. The \CPC algorithm can
therefore be used to solve a class of \emph{robust optimization problems in
  peer-to-peer processor networks}.

In particular, we consider constraint sets with parametric uncertainties of
the form
\begin{align} \label{prob.RobustOptimization} \mc{Z}_{i} = \{ z : \;
  f_{i}(z,\theta) \leq 0,\; \mbox{for\; all\;} \theta \in \Omega_{i} \}
\end{align}
where $\theta_{i}$ is an uncertain parameter, taking values in the compact
convex set $\Omega_{i}$.  We assume that $f_{i}$ is convex in $z$ for any fixed
$\theta$.  If additionally $f_{i}$ is concave in $\theta$ and $\Omega_{i}$ is a
convex set, we say that the resulting optimization problem \eqref{prob.Basic} is
convex \cite{Lopez2007}. As we will see later on, the first condition is
cruicial for the application of the algorithm. The second condition will ensure
that the problem can be solved exactly by our algorithm.

The problem \eqref{prob.Basic} with constraints of the form
\eqref{prob.RobustOptimization} is a \emph{distributed deterministic robust}
\cite{BenTal2009} or distributed \emph{semi-infinite} optimization problem
\cite{Lopez2007}. Each processor has knowledge of an infinite number of
constraints, determined by the parameter $\theta$ and the uncertainty set
$\Omega_{i}$. Obviously, uncertain constraints as
\eqref{prob.RobustOptimization} appear frequently in distributed decision
problems.
Here we focus on a deterministic worst-case optimization problem, where
a solution that is feasible for any possible representation of the
uncertainty is sought.
A comprehensive theory for robust optimization in centralized systems has been developed and is presented, e.g., in \cite{BenTal2009}.

Nowadays, mainly two different approaches are pursued in robust
optimization. In one research direction infinite, uncertain constraints are replaced by a
finite number of ``sampled" constraints. Sampling methods select a finite number
of parameter values and provide bounds for the expected violation of the
uncertain constraints \cite{Calafiore2010}. In a distributed setup, a sampling approach has been explored in \cite{Carlone2012a}.
The other research direction aims at formulating robust counterparts of
the uncertain constraints \eqref{prob.RobustOptimization}, leading often to
nominal semi-definite problems (see, e.g., \cite{BenTal2009}). Handling the uncertain constraint
from a semi-infinite optimization point of view \eqref{prob.RobustOptimization},
allows also to apply exchange methods \cite{Reemtsen1994}, where the sampling
point is chosen as the solution of a finite approximation of the optimization
problem. Recently, cutting-plane methods have been considered in the context of
centralized robust optimization \cite{Mutapcic2009}.
Robust optimization in processor networks is a relatively new problem. Robust
optimization for communication networks using dual decomposition is considered
in \cite{Yang2008}. We connect the robust optimization problem with uncertain
constraints \eqref{prob.RobustOptimization} to our general distributed
optimization framework, and show that the \CPC algorithm can solve the problem
in processor networks. In fact, the novel \CPC algorithm is related to the
exchange and cutting-plane methods \cite{Reemtsen1994}, \cite{Mutapcic2009}.
We define the cutting-plane oracle for the distributed robust optimization
problem \eqref{prob.RobustOptimization} as follows.

\begin{quote}
  \textbf{Pessimizing Cutting-Plane Oracle:} Given a query point $z_{q}$, the
  worst-case parameter value $\theta_{q}^{*}$ is the maximizer of the
  optimization problem
  \begin{align} \label{prob.PessimizingStep} \max_{\theta} &\;
    f_{i}(z_{q},\theta) \; \quad \mathrm{s.t.}\; \theta \in \Omega_{i}.
  \end{align}
  The query point $z_{q}$ is contained in $\mc{Z}_{i}$ if and only if the value
  of \eqref{prob.PessimizingStep} is smaller or equal to zero. If $z_{q} \notin
  \mc{Z}_{i}$, then cutting-plane is generated as
  \begin{align} \label{eqn.RobustCuttingPlane} f_{i}(z_{q},\theta_{q}^{*}) +
    g_{i}^{T}(z - z_{q}) \leq 0
  \end{align}
  where $g_{i}^{T} \in \partial f_{i}(z_{q},\theta_{q}^{*})$ is a subgradient of
  $f_{i}$.
\end{quote}
To see that \eqref{eqn.RobustCuttingPlane} is a cutting-plane, note that a query
point $z_{q} \notin \mc{Z}_{i}$ is cut off, since $f_{i}(z_{q},\theta_{q}^{*}) +
g_{i}^{T}(z_{q} - z_{q}) = f_{i}(z_{q},\theta_{q}^{*}) > 0$. Additionally, for
any point $z \in \mc{Z}_{i}$, we have $0 \geq f_{i}(z_{i},\theta)$ for all
$\theta \in \Omega_{i}$, and in particular $0 \geq f_{i}(z_{i},\theta_{q}^{*})
\geq f_{i}(z_{q},\theta_{q}^{*}) + g_{i}^{T}(z - z_{q})$.  Note that Assumption
\ref{ass.Separator} is satisfied since $f_{i}(z_{q},\theta_{q}^{*})=0$ implies
that $z_{q} \in \mc{Z}_{i}$.

The oracle of the robust optimization problem requires to solve an additional
optimization problem for determining the worst case parameter
\eqref{prob.PessimizingStep}. Following \cite{Mutapcic2009}, we call this the
\emph{pessimizing step}. For the practical applicability of our algorithm it is
important to stress that the pessimizing steps are performed in parallel on
different processors.

The pessimizing step can in general be performed by numerical tools. It can be
solved exactly if the problem is convex, i.e., $f_{i}$ is concave in the
uncertain parameter.

However, even if the convexity condition is not satisfied it might still be
possible to find an exact solution. Reference \cite{Mutapcic2009} provides a
formal discussion about when the pessimizing step can be solved exactly or even
analytically. We review here parts of the discussion.
Assume, e.g., that $f_{i}$ is convex in $\theta_{i}$ for all $z$, and
$\Omega_{i}$ is a bounded polyhedron, with the extreme points
$\{\theta_{i}^{1},\ldots,\theta_{i}^{k}\}$. The maximum of $f_{i}(z,u)$ is then
the maximum of $f_{i}(z,\theta_{i}^{1}),\ldots, f_{i}(z,\theta_{i}^{k})$, and
\eqref{prob.PessimizingStep} can be solved exactly by evaluating and comparing a
finite number of functions.
Furthermore, if $f_{i}(z,\theta_{i})$ is an affine function in $\theta_{i}$,
i.e., $f_{i}(z,\theta_{i}) = \alpha_{i}(z)\theta_{i} + \beta_{i}(z)$ and the
uncertainty set is an ellipsoid, i.e., $\Omega_{i}= \{\theta : \theta =
\bar{\theta}_{i} + P_{i}u, \; \|u\|_{2} \leq 1 \}$ for some nominal parameter
value $\bar{\theta}_{i}$ and a positive definite matrix $P_{i}$, then the
worst-case parameter value can be computed analytically as
\begin{align} \label{eqn.EllipsoidalWorstCase} \theta_{i}^{*} = \bar{\theta}_{i}
  + \frac{P_{i}P_{i}^{T}\alpha_{i}(z)}{\| P_{i} \alpha_{i}(z)\|_{2}}.
\end{align}
Finally, if $f_{i}$ is affine in the uncertain parameter and the uncertainty set
is a polyhedron, the pessimizing step \eqref{prob.PessimizingStep} becomes a
linear program.

\subsection*{Computational Study: Robust Linear Programming}

We evaluate in the following the time complexity of the algorithm in a
computational study for distributed robust linear programming. We follow here
\cite{Ben1999} and consider robust linear programs in the form
\eqref{prob.RobustOptimization} with linear uncertain constraints
\begin{align} \label{prob.RobustLP} a_{i}^{T}z \leq b_{i},\quad a_{i} \in
  \mc{A}_{i}, \quad i \in \{1,\ldots,n\}.
\end{align}
The data of the constraints is only known to be contained in a set, i.e., $a_{i}
\in \mc{A}_{i}$. Although our algorithm can in principle handle any convex
uncertainty set $\mc{A}_{i}$, we restrict us for this computational study to the
important class of \emph{ellipsoidal uncertainties} $\mc{A}_{i} = \{a_{i} :
a_{i} = \bar{a}_{i} + P_{i}u_{i}, \|u_{i}\|_{2} \leq 1 \}.$ The uncertainty
ellipsoids are centered at the points $\bar{a}_{i}$ and their shapes are
determined by the matrices $P_{i} \in \mathbb{R}^{d\times d}$.  It is known in
the literature that the centralized problem can be solved as a nonlinear
\emph{conic quadratic program} \cite{Ben1999}
\begin{align} \label{prob.RobustLP_SOCP}
  \begin{split}
    \max &\; c^{T}z, \quad \mbox{s.t.} \quad \bar{a}_{i}^{T}z + \|P_{i} z \|_{2}
    \leq b_{i}, \quad i \in \{1,\ldots,n\}.
  \end{split}
\end{align}
We will apply our algorithm directly to the uncertain problem model and use the
nonlinear problem formulation \eqref{prob.RobustLP_SOCP} only as a reference for
the computational study.
For the particular problem \eqref{prob.RobustLP} the pessimizing step can be
performed analytically. Note that $ \sup_{a_{i} \in \mc{A}_{i}} a_{i}^{T}z_{q} =
\bar{a}_{i}^{T}z_{q} + \sup_{\|u\|_{2} \leq 1} \{u^{T}P_{i}^{T}z_{q}\} =
\bar{a}_{i}^{T}z_{q} + \|P_{i}^{T}z_{q}\|_{2}.  $ The worst-case parameter is
therefore given by
\begin{align}
  a^{*}_{i} = \bar{a}_{i} + \frac{ P_{i}P_{i}^{T} z_{q} }{\|P_{i}z_{q}\|_{2}}.
\end{align}
A cutting-plane defined according to \eqref{eqn.RobustCuttingPlane} takes simply
the form $ a_{i}^{*}z \leq b_{i}, $ i.e., the linear constraint with the worst
case parameter value.

For the computational study, we generate random linear programs in the following
way.
The nominal problem data $a_{i} \in \mathbb{R}^{d}$ and $c \in \mathbb{R}^{d}$
are independently drawn from a Gaussian distribution with mean $0$ and standard
deviation $10$. The coefficients of the vector $b$ are then computed as $b_{i} =
\left( a_{i}^{T}a_{i} \right)^{1/2}$.  This random linear program model has been
originally proposed in \cite{Dunham1977}.
The matrices $P_{i}$ are generated as $P_{i}=M_{i}^{T}M_{i}$ with the
coefficients of $M_{i} \in \mathbb{R}^{d \times d}$ chosen randomly according to
a normal distribution with mean $0$ and standard deviation $1$.
All simulations are done with dimension $d = 10$. We consider the number of
communication rounds required until the query points of all processors are close
to the optimal solution $z^{*}$, i.e., we stop the algorithm centrally if for
all $i \in V$, $\|z^{[i]}(t)-z^{*}\|_{2} \leq 0.1$. In Figure
\ref{fig.RobustLP}, the completion time for two different communication graphs
is illustrated. We compare random Erd{\H{o}}s-R{\'e}nyi graphs, with edge
probability $p=1.2 \frac{\log(n)}{n}$, and circulant graphs with $5$
out-neighbors for each processor. It can be seen in Figure \ref{fig.RobustLP}
that the number of communication rounds grows with the network size for the
circulant graph, which have a growing diameter, but remains almost constant for
the random Erd{\H{o}}s-R{\'e}ny graphs, which have always a small diameter. The
simulations suggests, that the completion time depends primarily on the
\emph{diameter} of the communication graph.

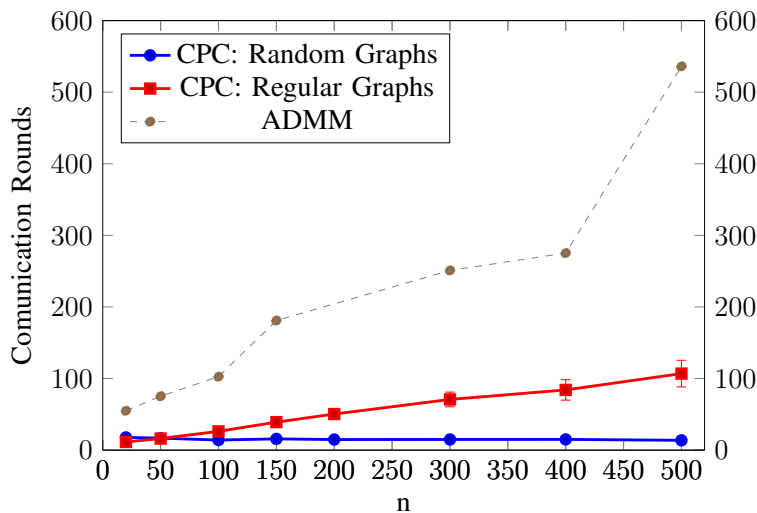
\begin{figure}[h!]
  \begin{center}
    \scalebox{0.9}{ \begin{tikzpicture}
	
    \begin{axis}[
     width=3.5in,
     height=2.5in,
        axis y line*=right,
		scale only axis,
		xmin=0,xmax=520,
		ymin=0, ymax = 600]
	%
	%
    \end{axis}	
	
	\begin{axis}[
	 width=3.5in,
     height=2.5in,
	    axis y line*=left,
	    xlabel = n,
		ylabel = Comunication Rounds,
		scale only axis,
		xmin=0,xmax=520,
		ymin=0, ymax=600,
		legend pos=north west]
		
		\addplot+[	line width=1.2pt,
		error bars/.cd, y dir=both,y explicit]
		coordinates{
		(20,17.8) +- (0,3.0103)
		(50,16.7) +- (0,1.7122)
		(100,14.1) +- (0,2.1651)
		(150,15.7) +- (0,1.3117)
		(200,14.78) +- (0,3.4394)
		(300,14.9) +-	(0,1.6026)
		(400,15) +- (0,2.3127)
		(500,13.7) +- (0,1.1476)
		};
		\addlegendentry{CPC: Random Graphs}
		\addplot+[	
		line width=1.2pt,
		error bars/.cd, y dir=both,y explicit]
		coordinates{
		(20,11.4) +- (0,1.9889)
		(50,16.2) +- (0,2.5174)
		(100,26.2) +- (0,2.7474)
		(150,39.1) +- (0,4.1958)
		(200,50.4) +- (0,5.9702)
		(300,70.8) +-	(0,10.2584)
		(400,84.1) +- (0,14.2777)
		(500,106.9) +- (0,18.5951)
		};
		\addlegendentry{CPC: Regular Graphs}

        \addplot+[color=gray, dashed]
		coordinates{
	    (20,55)
	    (50,75.3)
		(100,102.6)
		(150,181)
		(300,251)
		(400,275.25)
		(500,536)
		};
		\addlegendentry{ADMM}
	\end{axis}
\end{tikzpicture} }
    \caption{Average number of communication rounds and 95\% confidence interval
      required to compute the optimal solution to randomly generated robust
      linear programs with a precision of $\epsilon < 0.1$ for
      Erd{\H{o}}s-R{\'e}nyi graphs (blue) and circulant graphs (red) with the
      \CPC (CPC) algorithm. The dashed line shows for comparison the number of
      iterations of the ADMM algorithm with dual decomposition (dashed line).
    } \label{fig.RobustLP}
  \end{center}
\end{figure}

We consider for a comparison the ADMM algorithm combined with a
dual-decomposition, as described, e.g., in \cite[pp. 48]{Boyd2010}, to solve the
nominal conic quadratic problem representation \eqref{prob.RobustLP_SOCP} of the
robust optimization problem.\footnote{We use in the simulations a step-size
  $\rho = 200$, see \cite[Chapter 7]{Boyd2010} for the notation. Please note
  that the choice of the step-size of the ADMM method has to be done
  heuristically. We have selected the step-size as the best step-size we found
  experimentally for the smallest problem scenario $n=20$. Although the
  convergence speed of the ADMM method might improve with another step-size, in
  our experience most heuristic choices led to a significant deterioration of
  the performance.}
In one iteration of the ADMM algorithm, all processors must update their local
variables synchronously and then compute the average of all decision variables.
Figure \ref{fig.RobustLP} (right axis) shows the number of iterations of the
ADMM to compute the solution to the random linear programs with the same
precision as the \CPC algorithm. Note that the ADMM algorithm requires almost
three times more iterations than the \CPC algorithm requires communication
rounds.
Note also that the ADMM algorithm requires for each iteration an averaging of
the local solutions, which can be done by a consensus algorithm. Taking into
account that the number of communication rounds required to compute an average
by a consensus algorithm is lower bounded by
$\Omega\left(n^{2}\log(\frac{1}{\delta})\right)$, where $\delta$ is the desired
precision \cite{Olshevsky2009}, it is obvious that processors running the ADMM
algorithm need to communicate significantly more often than processors running
the \CPC algorithm.
Although the simulations do not compare the time-complexity of the
  algorithms in terms of computation units, they clearly suggest that the \CPC
algorithm is advantageous for applications where communication is costly or time
consuming.

\section{Separable Cost Optimization with Distributed Column
  Generation} \label{sec.SeparableCost}

The general convex problem set-up \eqref{prob.Basic} covers also the very
important class of \emph{almost separable optimization problems}, i.e., problems
where each processor is assigned local decision variables with a local objective
function and the local variables are coupled by a coupling constraint.
We sketch here the application of the \CPC algorithm to convex problems with
separable costs and linear coupling constraints of the form
\begin{align}
  \begin{split} \label{prob.SeparableCost}
    \min  & \quad \sum_{i=1}^{n} f_{i}(x_{i}) \\
    \mathrm{s.t.} &\quad \sum_{i=1}^{n} G_{i}x_{i} = \mathbf{h}, \quad x_{i} \in
    \mc{X}_{i},
  \end{split}
\end{align}
where $x_{i} \in \mathbb{R}^{m_{i}}$ is the decision vector assigned to
processor $i$, $f_{i}: \mathbb{R}^{m_{i}} \mapsto \mathbb{R}$ is a convex
objective function processor $i$ aims to minimize, and $\mc{X}_{i} \subset
\mathbb{R}^{m_{i}}$ is a convex set, defining the feasible region for the
decision vector $x_{i}$. For the clarity of presentation, we assume here that
all sets $\mc{X}_{i}$ are bounded, although this assumption can be relaxed.
The local decision variables $x_{i}$ are all coupled by a linear separable
constraint with a right-hand side vector $\mathbf{h} \in \mathbb{R}^{r}$. The
coupling linear constraint is of dimension $r$, and we assume here that $r$ is
small compared to the number of decision variables, i.e., $r \ll
\sum_{i=1}^{n}m_{i}$.

The problem formulation \eqref{prob.SeparableCost} is the standard formulation
considered for large scale optimization with decomposition methods
\cite{Lasdon2002}. Standard large-scale optimization methods for
\eqref{prob.SeparableCost} exploit the separable structure of the dual problem,
and define a coordinating master program and several sub-problems, leading to a
structure as shown in Figure \ref{fig.DecompositionStructures}(a). In contrast,
we are seeking an optimization method without a master problem using only
asynchronous message-passing between neighboring processors, as visualized in
Figure \ref{fig.DecompositionStructures}(b).

The method we propose here is strongly related to the classical
\emph{Dantzig-Wolfe (DW) decomposition} or \emph{column generation}
\cite{Dantzig1961}, \cite{Lasdon2002}. The DW decomposition is dual to the
cutting-plane method, see e.g., \cite{Eaves1971}. We exploit this duality
relation here. Once again we want to stress that the DW decomposition requires a
coordinating master problem, which is not required for our algorithm.  In
\cite{Burger2011b} we proposed a similar algorithm for purely linear programs taking only the primal perspective on the problem.

The problem \eqref{prob.SeparableCost} can be formulated in the general
framework \eqref{prob.Basic}, when its dual is considered. Let $\pi \in
\mathbb{R}^{r}$ be the dual variable corresponding to the coupling
constraint. The dual problem to \eqref{prob.SeparableCost} can then be written
as
\begin{align*}
  \begin{split}
    \max_{\pi} \; -\mathbf{h}^{T}\pi + \sum_{i=1}^{n} \left\{ \min_{x_{i} \in
        \mc{X}_{i}} f_{i}(x_{i}) + \pi^{T}G_{i}x_{i} \right\}.
  \end{split}
\end{align*}
One can now define a new variable $u_{i} := \min_{x_{i} \in \mc{X}_{i}}
f_{i}(x_{i}) + \pi^{T}G_{i}x_{i},$ leading to the alternative representation of
the dual as
\begin{align} \label{prob.ColumnDual}
  \begin{split}
    \max_{\pi,u_{i} } \; &-\mathbf{h}^{T}\pi + \sum_{i=1}^{n} u_{i} \\
    &(\pi,u) \in \{(\pi,u): u_{i} \leq f_{i}(x_{i}) + \pi^{T}G_{i}x_{i},\;
    \forall x_{i} \in \mc{X}_{i} \}.
  \end{split}
\end{align}
This problem is explicitly in the form \eqref{prob.Basic} with $z = [\pi^{T},
u_{1},\ldots,u_{n}]^{T} \in \mathbb{R}^{r+n}$, $c = [-\mathbf{h}^{T}, \1^{T}_{n}
]^{T}$ and $ \mc{Z}_{i} := \{ (\pi,u_{i}) : u_{i} \leq f_{i}(x_{i}) +
\pi^{T}G_{i}x_{i}, \forall x_{i} \in \mc{X}_{i} \}.  $ The cutting-plane oracle
can now be defined as follows. A query point is denoted as $z_{q} =
[\pi_{q}^{T}, u_{q,1}, \ldots, u_{q,n}]^{T}$ and is contained in the set
$\mc{Z}_{i}$ if and only if
\begin{align*}
  u_{q,i} \leq f_{i}(x_{i}) + \pi_{q}^{T}G_{i}x_{i},\quad \forall x_{i} \in
  \mc{X}_{i}.
\end{align*}
\begin{quote} \textbf{Constraint Generating Oracle:} Let $\bar{x}_{i}$ denote
  the optimal solution vector to
  \begin{align} \label{prob.SubProblem} \min_{x_{i}} \; f_{i}(x_{i}) +
    \pi_{q}^{T}G_{i}x_{i},\quad \mathrm{s.t.\;} x_{i} \in \mc{X}_{i}
  \end{align}
  and let $\gamma_{i}^{*}$ be the optimal value of \eqref{prob.SubProblem}.
  If $u_{q,i} > \gamma_{i}^{*} $ then $z_{q} \notin \mc{Z}_{i}$.
  A cutting plane separating $z_{q}$ and $\mc{Z}_{i}$ is then
  \begin{align}
    u_{i} - f_{i}(\bar{x}_{i}) - \pi^{T}G_{i}\bar{x}_{i} \leq 0.
  \end{align}
\end{quote}
Clearly, $u_{q,i} - f_{i}(\bar{x}_{i}) - \pi_{q}^{T}G_{i}\bar{x}_{i} >
0$ for $ (\pi_{q}, u_{q}) \notin \mc{Z}_{i}$ and $u_{q,i} - f_{i}(\bar{x}_{i}) -
\pi_{q}^{T}A_{i}\bar{x}_{i} \leq 0$ for all $(\pi,u) \in \mc{Z}_{i}$.
Also, Assumption \ref{ass.Separator} holds since $s(z_{q}) = u_{q,i} -
f_{i}(\bar{x}_{i}) - \pi_{q}^{T}G_{i}\bar{x}_{i}$ and $s(z_{q}) \rightarrow 0$
implies $(\pi_{q},u_{q}) \in \mc{Z}_{i}$.

The proposed procedure of constructing a constraint is known as ``constraint
generation" or, taking the primal perspective, as ``column generation".
We name \eqref{prob.SubProblem} the \emph{local subproblem} $SP_{i}$, since it
corresponds to the subproblem of the DW decomposition.
The approximate linear program formed by each processor is called here
\emph{local master problem} $MP_{i}$, since it is a local version of the master
program of the DW-decomposition. 

It is worth noting that here $z = [\pi^{T}, u_{1},\ldots,u_{n}]^{T}$ and thus
the dimension of the problem, $d = r+ n$, is no longer independent of the number
of processors. Additionally, the set-up considered in this section requires a
unique identifier to be assigned to each processor. These two additional
restrictions have to be taken into account for an implementation of the
algorithm.
\begin{figure}[t]
  \begin{center}
    \subfigure[Structure of the classical Dantzig-Wolfe
    decomposition.]{ \includegraphics[width=0.4\textwidth, trim=4cm 18cm 3cm 4cm]{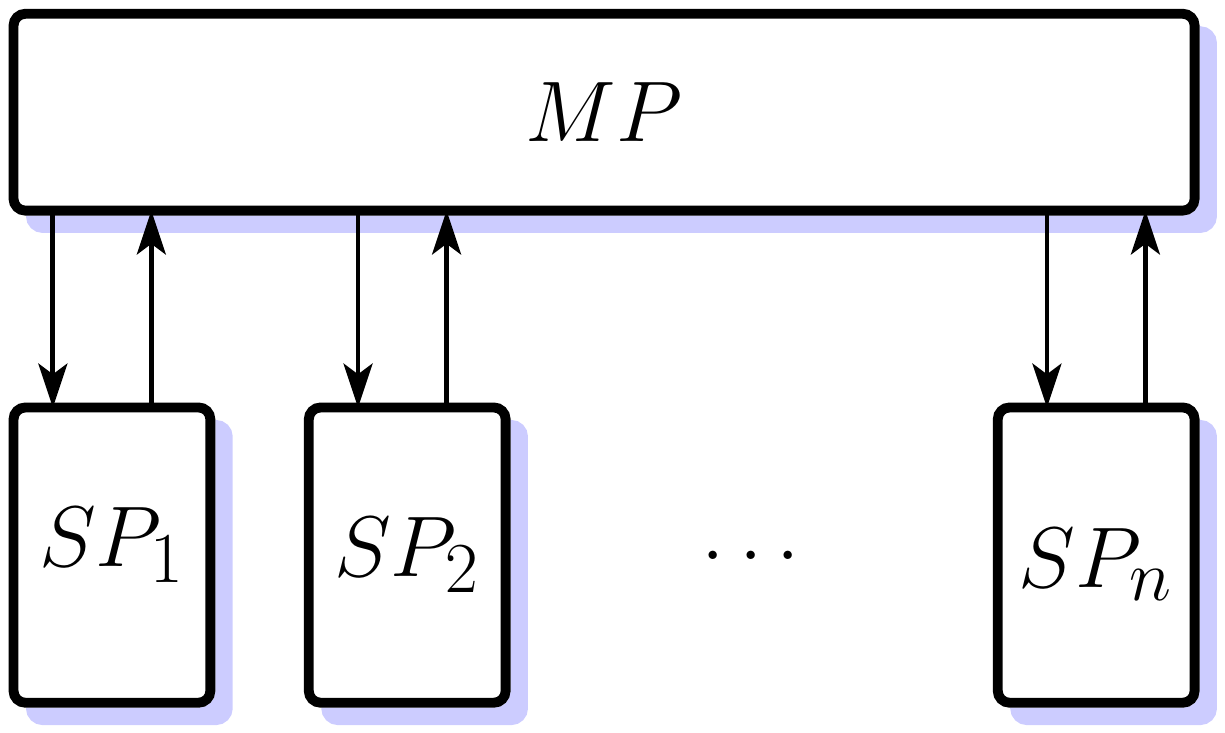}} \label{fig.DecompositionStructures_DW}
 \hfill
    \subfigure[Structure of the \CPC algorithm for separable
    problems.]{\includegraphics[width=0.45\textwidth, trim=4cm 17.5cm 0.5cm 5cm]{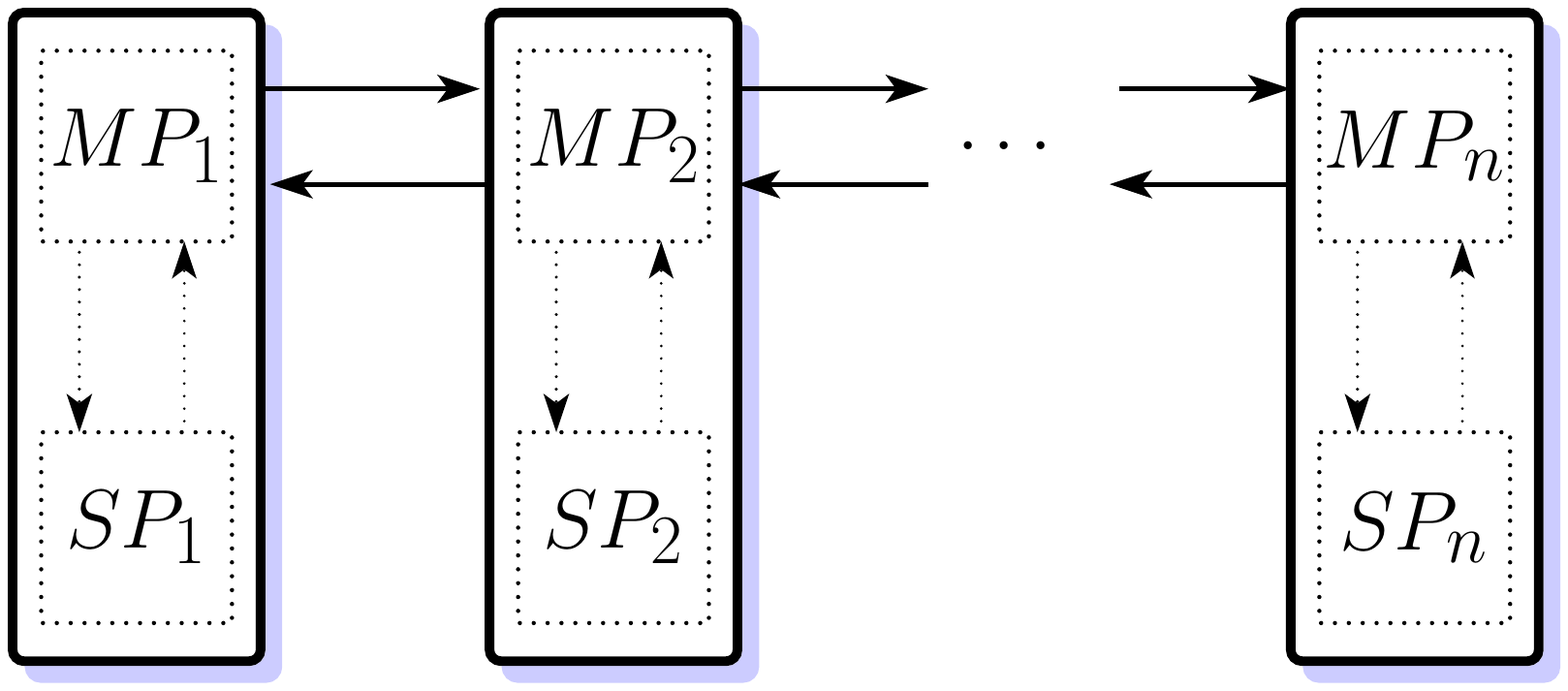}}
    \label{fig.DecompositionStructures_CPC}
    \caption{Comparison of the classical master / subproblem structure of the DW
      decomposition and the peer-to-peer structure of the \CPC
      algorithm.} \label{fig.DecompositionStructures}
  \end{center}
\end{figure}

The \CPC algorithm is applied here to the dual problem and will compute the dual
solution to \eqref{prob.SeparableCost}, i.e., $$\lim_{t\rightarrow \infty}
\|\pi^{[i]}(t) -\pi^{*}\|_{2} \rightarrow 0. $$
If all $f_{i}(\cdot)$ in \eqref{prob.SeparableCost} are strictly convex, the
solutions of the local subproblems \eqref{prob.SubProblem} of each processor
will converge to the optimal solution, i.e., $\lim_{t \rightarrow \infty}
\|\bar{x}_{i}^{[i]}(t) - x_{i}^{*}\| \rightarrow 0, \; \mbox{for\;all\;} i \in
V, $ where $x^{*}=[x_{1}^{*},\ldots,x_{n}^{*}]$ is the optimal primal solution
to \eqref{prob.SeparableCost}, and $\bar{x}_{i}^{[i]}(t)$ is the solution to
\eqref{prob.SubProblem} computed by processor $i$ at time $t$.
However, this is not true if some $f_{i}(\cdot)$ are only convex but not
strongly convex. Then recovering a primal optimal solution from the dual
solution can be done using the method known from DW decomposition.
We assume that each processor stores the points at which a constraint is
generated, $\bar{x}^{[i]}(\tau)$, where the index $i$ indicates which processor
computed at time $\tau$ the point $\bar{x}_{i}(\tau)$ as solution to
\eqref{prob.SubProblem}.
Define $\bar{G}_{i\tau} := G_{i}\bar{x}_{i}(\tau)$ and $\bar{f}_{i\tau} :=
f_{i}(\bar{x}_{i}(\tau))$.
The scalar inequalities of the approximate linear program are all of the form
\begin{align} \label{eqn.GeneratedConstraint} u_{i} -\bar{G}_{i\tau}^{T} \pi
  \leq \bar{f}_{i\tau}.
\end{align}
One can now formulate the linear programming dual to the approximate program
\eqref{prob.BasicApproxLP}. Let $\lambda_{j\tau} \in \mathbb{R}_{\geq 0}$ be the
Lagrange multiplier to the constraint \eqref{eqn.GeneratedConstraint}, the
linear programming dual to \eqref{prob.BasicApproxLP} is a linear program with
the following structure:
\begin{align} \label{prob.ApproxDualLP}
  \begin{split}
    \min_{\lambda_{i\tau} \geq 0}\; &\sum_{i=1}^{n} \sum_{\tau} \bar{f}_{i\tau} \lambda_{i\tau} \\
    &\sum_{i=1}^{n} \sum_{\tau} \bar{G}_{i\tau}\lambda_{i\tau} = \mathbf{h},
    \quad \sum_{\tau} \lambda_{i\tau} = 1, \; i \in \{1,\ldots,n\}.
  \end{split}
\end{align}

We assume in the following that all processors have the same set of constraints
\eqref{eqn.GeneratedConstraint} as their basis. Please note that this can be
achieved by halting the algorithm at some time and running a suitable agreement
mechanism, such as the one proposed in \cite{Notarstefano2009}.
A processor can now reconstruct its component of the solution vector as the
convex combination $ x_{i}^{*} = \sum_{\tau}
\bar{x}_{i}(\tau)\lambda_{i\tau}^{*},$ where $\lambda_{i\tau}^{*}$ solves
\eqref{prob.ApproxDualLP}.
The resulting solution vector $x^{*} = [ x_{1}^{*},\ldots, x_{n}^{*}]$ is
globally feasible since $ \sum_{i=1}^{n} \sum_{\tau}
\bar{G}_{i\tau}\lambda_{i\tau}^{*} = \sum_{i=1}^{n} G_{i} \left(\sum_{\tau}
  \bar{x}_{i}(\tau)\lambda_{i\tau}^{*} \right) = \sum_{i=1}^{n}G_{i}x_{i}^{*}=
\mathbf{h}.  $
Additionally, if all processors have computed the globally optimal solution to
\eqref{prob.ColumnDual}, then the recovered $x_{i}^{*}=\sum_{k}
\bar{x}_{i}(\tau)\lambda_{i\tau}$ is also the optimal primal solution to
\eqref{prob.SeparableCost}.
To see this note that strong duality implies that the optimal value of
\eqref{prob.ApproxDualLP} is equivalent to the value of the linear approximate
problem \eqref{prob.BasicApproxLP}, which we denote with $f^{*}$. Thus, $f^{*}=
\sum_{i=1}^{n} \sum_{\tau=1} f_{i}(\bar{x}_{i}(\tau))\lambda_{i\tau}^{*}$.
Convexity of $f_{i}(\cdot)$ and $\sum_{\tau} \lambda_{i\tau} = 1$ implies that $
f^{*} = \sum_{i=1}^{n} \sum_{\tau} f_{i}(\bar{x}_{i}(\tau))\lambda_{i\tau}^{*}
\geq \sum_{i=1}^{n} f_{i}(\sum_{\tau} \bar{x}_{i}(\tau) \lambda_{i\tau}^{*}) =:
\sum_{i=1}^{n} f_{i}(x_{i}^{*}).  $ Since $x^{*}=[x_{1}^{*},\ldots,x_{n}^{*}]$
is a feasible solution it must hold that $\sum_{i=1}^{n} f_{i}(x_{i}^{*}) =
f^{*}$.
Please note that the proposed method requires each processor to store its own
local solutions $\bar{x}_{i}^{[i]}(t)$ to \eqref{prob.SubProblem} generated
during the evolution of the algorithm, but does not require that the processors
exchange those solutions. For a more explicit discussion on the reconstruction
of the feasible solution, we refer the reader to the literature on nonlinear
DW-decomposition \cite{Lasdon2002} or our recent paper \cite{Burger2011b}.

\subsection*{Application Example: Distributed Microgrid Control }
The previous discussion shows that the \CPC algorithm is applicable for many
important control problems, such as for example distributed microgrid control.
Microgrids are local collections of distributed energy sources, energy storage
devices and controllable loads.
Most existing control strategies still use a central controller to optimize the
operation \cite{Zamora2010}, while for several reasons, detailed, e.g., in
\cite{Zamora2010}, distributed control strategies, which do not require to
collect all data at a central coordinator, are desirable.

We consider the following optimization model of the microgrid, described
recently in \cite{Kraning2012}. A microgrid consists of several generators,
controllable loads, storage devices and a connection to the main grid over which
power can be bought or sold. In the following, we use the notational convention
that energy generation corresponds to positive variables, while energy
consumption corresponds to negative variables.
A \emph{generator} generates power $p_{gen}(t), t \in [0,T]$ within the absolute
bounds $\underline{p}(t) \leq p_{gen}(t) \leq \bar{p}(t)$ and the rate
constraints $\underline{r}(t) \leq p_{gen}(t+1) - p_{gen}(t) \leq
\bar{r}(t)$. The cost to produce power by a generator is modeled as a quadratic
function $f_{gen}(t) = \alpha p_{gen}(t) + \beta p^2_{gen}(t)$.
A \emph{storage device} can store or release power $p_{st}(t), t \in [0,T]$
within the bounds $-d_{st} \leq p_{st}(t) \leq c_{st}$. The charge level of the
storage device is then $ q_{st}(t) = q_{st, init} +
\sum_{\tau=0}^{t}p_{st}(\tau)$ and must be maintained between $0 \leq q_{st}(t)
\leq q_{max}$. Note that $p_{st}(t)$ takes negative values if the storage
device is charged and positive values if it is discharged.
A \emph{controllable load} has a desired load profile $l_{cl}(t)$ and
incorporates a cost if the load is not satisfied, i.e., $f_{cl}(t) =
\alpha(l_{cl}(t)-p_{cl}(t))_{+}$, where $(z)_{+} = \max\{0,z\}$.
Finally, the microgrid has a single control unit, which coordinates the
connection to the main grid and can trade energy. The maximal energy that can be
traded is $|p_{tr}| \leq E$. The cost to sell or buy energy is modeled as $
f_{tr} = -c^{T}p_{tr} + \gamma^{T} |p_{tr}|$ where $c$ is the price vector and
$\gamma$ is a general transaction cost.

The power demand $D(t)$ in the microgrid is predicted over a horizon $T$. The
control objective is to minimize the cost of power generation while satisfying
the overall demand.  This control problem can be directly formulated as in the
form \eqref{prob.SeparableCost}, with
the local objective functions $f_{i} = \sum_{t = 0}^{T}f_{i}(t)$,
the right-hand side vector of the coupling constraint as the predicted demand
$\mathbf{h}= [D(1), \ldots, D(T)]^{T}$ and $\mathcal{X}_{i}$ as the local
constraints of each unit.

The \CPC algorithm can solve this problem in a distributed way. Note that the
objective functions $f_{i}$ considered here are all convex, but not strictly
convex. If all objective functions were strictly convex, one could use the
distributed Newtons method \cite{Zargham2011}, which has locally a quadratic
convergence rate. However, the distributed Newton method does not apply to this
problem formulation. The \CPC algorithm does not require strict convexity of the
cost functions.

We present simulation results for an example set-up with $n=101$ decision units,
i.e., 60 generators, 20 storage devices, 20 controllable loads and one
connection to the main grid. A random demand is predicted for 15 minute time
intervals over a horizon of three hours, based on a constant off-set, a
sinusoidal growth and a random component.
The algorithm is initialized with each processor computing a basis out of the
box-constraint set $\{z : -10^{5} \cdot \1 \leq z \leq 10^{5}\cdot \1 \}$,
leading to a very high initial objective value.
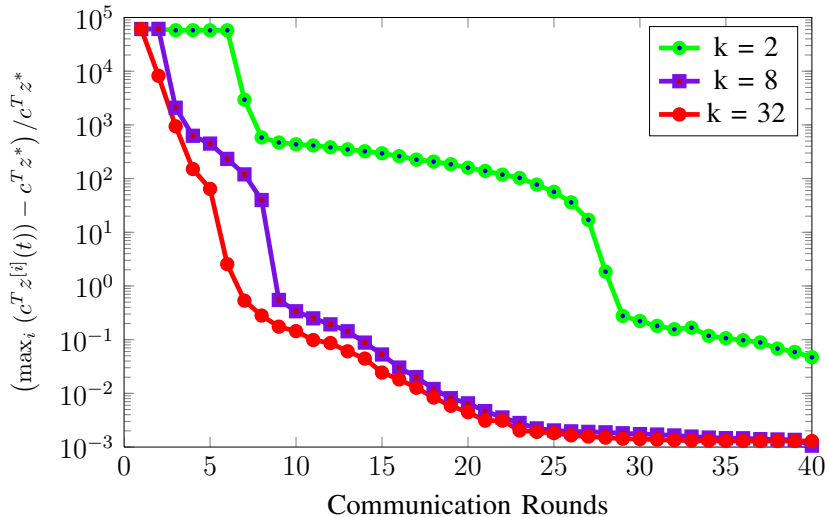
\begin{figure}[h!]
  \begin{center}
     \scalebox{0.9}{ 
%
%
%
%
\begin{tikzpicture}

\definecolor{mycolor1}{rgb}{0.47843137254902,0.0627450980392157,0.894117647058824}

\begin{semilogyaxis}[%
width=4in,
height=2.5in,
scale only axis,
xmin=0, xmax=40,
xlabel={Communication Rounds},
ymin=0.001, ymax=100000,
yminorticks=true,
ylabel={{\small $\left(\max_{i}\; (c^{T}z^{[i]}(t)) -c^{T}z^{*}\right)/c^{T}z^{*}$}},
axis on top]
\addplot+ [
color=green,
solid,
line width=2.0pt
]
coordinates{
 (1,61041.810800659)(2,60916.2449136258)(3,57878.6936370382)(4,58025.6893711037)(5,57795.628521469)(6,57789.2423779345)(7,2953.80995874)(8,582.504467829637)(9,469.569036543717)(10,435.438949695347)(11,414.037861447156)(12,382.939497528563)(13,349.790241058142)(14,323.21682647535)(15,296.425407793386)(16,261.492639203813)(17,225.505292386678)(18,207.575696253664)(19,185.287018120244)(20,159.635777400795)(21,139.085670834991)(22,118.452386035394)(23,102.980480409093)(24,77.3120499957826)(25,56.7336076376978)(26,36.1385614292732)(27,17.1025633985353)(28,1.84769558741159)(29,0.276700896843628)(30,0.223010142467051)(31,0.180721728729839)(32,0.156720877351767)(33,0.167259745732316)(34,0.118077100963003)(35,0.106550140220224)(36,0.0977627375565929)(37,0.0889385157178269)(38,0.0682922376902758)(39,0.0590736151494529)(40,0.0470325505277574) 
};
\addlegendentry{k = 2}

\addplot+ [
color=mycolor1,
solid,
line width=2.0pt
]
coordinates{
 (1,61041.8134388754)(2,61378.4104113139)(3,2087.43598991165)(4,627.14830721006)(5,448.772502527723)(6,232.791660748588)(7,119.829990361373)(8,39.9122182721396)(9,0.549643025081249)(10,0.338165947227444)(11,0.249553863774701)(12,0.193129018721244)(13,0.143359263788978)(14,0.0883164915868461)(15,0.053036588399771)(16,0.0302490829895319)(17,0.0200616724885256)(18,0.0120704927907592)(19,0.00813150860548946)(20,0.00652119561169848)(21,0.0046678566797171)(22,0.00356178465474531)(23,0.00280610425412551)(24,0.0022298121226671)(25,0.0020405452649379)(26,0.00196352552972864)(27,0.00191918982562351)(28,0.00188331371732702)(29,0.00180695488208189)(30,0.00174257501747803)(31,0.00171039184326654)(32,0.00164566164797043)(33,0.0015673043251899)(34,0.00150423366501222)(35,0.00146440522261942)(36,0.00145879756886379)(37,0.00140338632596726)(38,0.00138561843977964)(39,0.00136110501498539)(40,0.00106173546346672) 
};
\addlegendentry{k = 8}
\addplot+ [
color=red,
solid,
line width=2.0pt
]
coordinates{
 (1,61041.8204440049)(2,8214.95259187378)(3,941.55688224359)(4,150.623584032368)(5,64.3678322561843)(6,2.55010524133251)(7,0.533864396813393)(8,0.28146634615064)(9,0.174632496601746)(10,0.144214212473858)(11,0.0984714410426233)(12,0.08662048999983)(13,0.060725304459828)(14,0.0444374807842885)(15,0.0243988076258519)(16,0.0181718600637782)(17,0.0127005584686388)(18,0.0083612382934639)(19,0.00585946649018298)(20,0.00445641841938799)(21,0.0031008579411624)(22,0.0031008579411624)(23,0.00203727469521974)(24,0.00191791739761034)(25,0.0018208542200473)(26,0.00165243066919092)(27,0.00156982812760242)(28,0.00149836767256575)(29,0.00143018487066668)(30,0.00140765492179307)(31,0.00137295391672404)(32,0.00133832169718873)(33,0.00132370882938308)(34,0.00131184162306753)(35,0.00129507524724541)(36,0.00129507524724541)(37,0.00129507524724541)(38,0.00129507327069365)(39,0.00129507524724644)(40,0.00129507524724541) 
};

\addlegendentry{k = 32}
\end{semilogyaxis}
\end{tikzpicture} }
    \caption{Trajectories of the scaled maximal optimal value of the linear
      approximate programs for different $k$-regular communication graphs
      $\mc{G}_{c}$.  } \label{fig.Microgrid}
  \end{center}
\end{figure}
Figure \ref{fig.Microgrid} shows the largest objective value over all
processors, relative to the best solution found as the algorithm is continued to
perform.  The evolution of the objective value is shown for three different
$k$-regular graphs.
It can be clearly seen that the convergence speed depends strongly on the
structure of the communication graph. The convergence for a network with a
2-regular communication structure is significantly slower than for a network
with a higher regular graph. We also want to emphasize the observation that the
difference in the convergence speed between $k=8$ and $k=32$ is not as big as
the increased communication would let one expect. This shows that the
improvement obtained from more communication between the processors becomes
smaller with more communication. A good performance of the algorithm can also be
obtained with little communication between the processors.
Please note that for all communication graphs the \CPC algorithm requires only
few communication rounds to converge to a fairly good solution.
Although the convergence to an exact optimal solution might take more
iterations, a good sub-optimal solution can be found after very few
communication rounds. This property makes the \CPC attractive for control and
decision applications.

\section{Discussion and Conclusions} \label{sec.Conclusions}

We proposed a framework for distributed convex and robust optimization using a
polyhedral approximation method.
As a general problem formulation, we consider problems where convex constraint
sets are distributed to processors, and the processors have to compute the
optimizer of a linear objective function over the intersection of the constraint
sets.
We proposed the novel \CPC algorithm as an asynchronous algorithm performing in
peer-to-peer networks. The algorithm is well scalable to large networks in the
sense that the amount of data each processor has to store and process is small
and independent of the network size.

The appealing property of the considered outer-approximation method lies in the
fact that it imposes very little requirements on the structure of the constraint
sets.
Merely the only requirement is that a cutting-plane oracle exists. We have
presented oracles for various formulations of the constraint sets, in
particular, inequality and convex uncertain or semi-infinite
constraints. Also, we showed that, as the dual problem formulation is
considered, also almost separable convex optimization problems can be formulated
in the proposed framework.
We showed for each of the proposed problem formulations how the cutting-plane
oracle can be defined.

Finally, we illustrated that the proposed set-up is of
interest for various decision and control problems. These include the
localization problem in sensor networks.
They include also less obvious problems as, e.g., distributed microgrid control,
where the novel algorithm can be applied to the dual problem formulation. In
this context we showed that the application of the algorithm to the dual problem
has the major advantage that a feasible solution can be found in a fully
distributed way even before the algorithm has converged to an optimal solution.

\appendix

\setcounter{section}{9}

\subsection{Proofs of Section \ref{sec.OuterApproximation}}

\subsubsection{Proof of Proposition \ref{prop.2normComputation} }
The minimal 2-norm solution is the solution to 
\begin{align}\label{prob.MinNormQP}
\min_{z,y} \frac{1}{2}z^{T}z, \; \mbox{s.t.} \; A_{H}^{T}z \leq b_{H}, \; A_{H}y = c, \; c^{T}z - b_{H}^{T}y = 0,\; y \geq 0,
\end{align}
where the constraints represent the linear programming optimality conditions (KKT-conditions).
The Lagrangian of \eqref{prob.MinNormQP} can be directly determined to be
\begin{align}
\mathcal{L}(z,y,u,l,\alpha) = \frac{1}{2}z^{T}z + u^{T}(A_{H}^{T}z -b_{H}) + l^{T}(A_{H}y-c) + \alpha(c^{T}z-b_{H}^{T}y),\; y,u \geq 0.
\end{align}
It follows now that $y^{*} = \mbox{arg}\min_{y\geq 0} \; \mathcal{L}(z,y,u,l,\alpha) = 0$ if 
$ A_{H}^{T}l - \alpha b_{H} \geq 0.$ 
From $ z^{*} = \mbox{arg} \min_{z}\;, \mathcal{L}(z,y,u,l,\alpha) $ follows that $z^{*} = -A_{H}u - \alpha c.$
The problem \eqref{prob.2NormDual} stated in the proposition is now $\min_{u \geq 0, l, \alpha} \; -\mathcal{L}(z^{*},y^{*},u,l,\alpha)$.  $\hfill \blacksquare$

\vspace{1em}

\subsubsection{Proof of Lemma \ref{prop.QuadraticObjective}}

The minimal 2-norm solution $z_{H}^{*}$ is the unique minimizer of 
\begin{align*}
\begin{split}
\min_{z} \; & \frac{1}{2}\|z\|^{2},\quad \mathrm{s.t.}\; c^{T}z \geq \gamma_{H},\; A_{H}^{T}z \leq b_{H}.
\end{split}
\end{align*}
and satisfies therefore the feasibility conditions
$
c^{T}z^{*}_{H} = \gamma_{H}$ and $A_{H}^{T}z^{*}_{H} \leq b_{H}.$
Since $z_{H}^{*}$ is an optimal solution, there exist multipliers $\mu^{*} \in \mathbb{R}$ and $\lambda^{*} \in \mathbb{R}^{|H|}_{\geq 0}$, such that the KKT conditions are satisfied, i.e.,
\begin{align}
z^{*}_{H} - \mu^{*}c + A_{H} \lambda^{*} &= 0 \label{eq.FOC_QP_1} \\
\lambda^{*T}A_{H}^{T}z^{*}_{H} - \lambda^{*T}b_{H} &= 0. \label{eq.FOC_QP_2} 
\end{align}
Since $z^{*}_{H}$ is also a solution to the original linear program  \eqref{prob.BasicApproxLP}, there also exists a multiplier vector $y^{*} \in \mathbb{R}^{|H|}_{\geq 0}$ satisfying the linear programming optimality conditions
\begin{align}
-c + A_{H}y^{*} &= 0  \label{eq.FOC_LP_1}\\
y^{*T}A_{H}^{T}z^{*}_{H} - b^{T}_{H}y^{*} &= 0. \label{eq.FOC_LP_2}
\end{align} 

We have to show now that the existence of $z^{*}_{H},\mu^{*},\lambda^{*}$ and $y^{*}$ imply, for a sufficiently small $\epsilon$, the existence of a multiplier vector $\pi^{*}$ satisfying the optimality conditions of \eqref{prob.QuadraticPerturbation}, which are
\begin{align}
\begin{split}
-c + \epsilon z^{*}_{H} + A_{H}\pi^{*} &= 0\quad \mbox{and} \quad
\pi^{*T}A_{H}^{T}z^{*}_{H} - \pi^{*}b_{H} = 0. \label{eq.FOC_Perturbed}
\end{split}
\end{align}

We distinguish now the two cases $\mu^{*} > 0$ and $\mu^{*} = 0$. First, assume $\mu^{*} > 0$. We can multiply \eqref{eq.FOC_LP_1} with $\frac{t}{\mu^{*}}$, for arbitrary $t \in (0,1]$, and add to this \eqref{eq.FOC_LP_1}, multiplied by $(1-t)$ to obtain
\begin{align}
\frac{t}{\mu^{*}}z^{*}_{H} - c + A_{H}(\frac{t}{\mu^{*}}\lambda^{*} + (1-t)y^{*}) = 0. \label{eq.MUpos_1}
\end{align}
The same steps can be repeated with \eqref{eq.FOC_QP_2} and \eqref{eq.FOC_LP_2} to obtain
\begin{align}
(\frac{t}{\mu^{*}}\lambda^{*T} + (1-t)y^{*T})(A_{H}^{T}z^{*}_{H} - b_{H}) = 0. \label{eq.MUpos_2}
\end{align}
With \eqref{eq.MUpos_1} and \eqref{eq.MUpos_2}, for any $\epsilon \leq \frac{1}{\mu^{*}}$, one can define $t_{\epsilon} = \epsilon \mu^{*}$. Then 
$
\pi^{*} = \frac{t_{\epsilon}}{\mu^{*}}\lambda^{*T} + (1-t_{\epsilon})y^{*T}
$
solves \eqref{eq.FOC_Perturbed}.
In the second case $\mu^{*}=0$, one can pick an arbitrary $\epsilon >0$, multiply \eqref{eq.FOC_QP_1} (and \eqref{eq.FOC_QP_2}, respectively) with $\epsilon$ and add \eqref{eq.FOC_LP_1} (or \eqref{eq.FOC_LP_2}, respectively) to obtain
$
\epsilon z^{*}_{H} - c + A_{H}(\epsilon \lambda^{*} +y^{*}) = 0$ and $
(\epsilon \lambda^{*} +y^{*})(A_{H}^{T}z^{*}_{H} - b_{H}) = 0.$
Now, $\pi^{*} := (\epsilon \lambda^{*} +y^{*})$ solves \eqref{eq.FOC_Perturbed}. 
$\hfill \blacksquare$

%
%
%
%
%
%

\subsection{Proofs of Section IV: Correctness of the Algorithm} \label{sec.AppendixProofs}

Some technical properties of the algorithm are formalized in the following result. 

\begin{applemma} \label{prop.TechIssues} Let $z^{[i]}(t)$ be the query point and
  $B^{[i]}(t)$ the corresponding basis. Let $\Bc^{[i]}(t)
  \subset \mathbb{R}^{d}$ be the feasible set induced by
  $B^{[i]}(t)$. Then,
  \begin{enumerate}
  \item $\Bc^{[i]}(t) \supset \mc{Z}$ for all $i \in \{1\ldots,n\}$ and $ t \geq
    0$;
  \item $\lim_{t \rightarrow \infty} z^{[i]}(t) = \bar{z}$ and $\bar{z} \in
    \mc{Z}$ implies $\bar{z}$ is a minimizer of \eqref{prob.Basic};
  \item there exists $\underline{\epsilon} > 0$ such that for all $i \in
    \{1,\ldots,n\}$ and all $t \geq 0$, the query points $z^{[i]}(t)$ maximize
    the objective function
 $$\Jeps(z) := c^{T}z - \frac{\epsilon}{2}\|z\|^{2}_{2}$$ 
 over the set of constraints $B^{[i]}(t) \cup Y^{[i]}(t)$ (as defined in (S2))
 for all $\epsilon \in [0,\underline{\epsilon}]$;
\item $\Je(z^{[i]}(t+1)) \leq \Je(z^{[i]}(t))$ for all $i \in \{1,\ldots,n\}$
  and all $t\geq 0$;
\item if $\mc{G}_{c}$ is a strongly connected \emph{static} graph, then
$\Je(z^{[j]}(t+\diam(\mc{G}_{c})) ) \leq \Je(z^{[i]}(t))$
for all $i,j \in \{1,\ldots,n\}$ and all $t\geq0$.
\end{enumerate}
\end{applemma}

\begin{proof}
  To see (i), note that any cut $h_{k}$ generated by the oracle of processor
  $i$, $\orc(\cdot,\mc{Z}_{i})$ is such that the half-space $h_{k}$ contains
  $\mc{Z}_{i}$, and in particular $h_{k}$ contains $\mc{Z} =
  \bigcap_{i=1}^{n}\mc{Z}_{i}$. Thus any collection of cuts $H = \bigcup_{k}
  h_{k}$, generated by arbitrary processors is such that $\mc{H} \supset
  \bigcap_{i=1}^{n} \mc{Z}_{i} = \mc{Z}$, and in particular $\mc{B}^{[i]}(t)
  \supset \mc{Z}$.
  The claim (ii) follows since $z^{[i]}(t)$ is computed as a maximizer of the
  linear cost $c^{T}z $ over the collection of cutting-planes
  $H^{[i]}_{tmp}(t)$. The induced polyhedron is such that $\Hc_{tmp}^{[i]}(t)
  \supset \mc{Z}$. Therefore, we can conclude that $c^{T}z^{[i]}(t) \geq
  c^{T}z^{*}$, where $z^{*}$ is an optimizer of \eqref{prob.Basic}. By
  continuity of the linear objective function, we have that $c^{T}\bar{z}\geq
  c^{T}z^{*}$ On the other hand, $c^{T}z \leq c^{T}z^{*}$ for all $z \in
  \mc{Z}$. This proves the statement.
  The statement (iii) follows from Lemma
  \ref{prop.QuadraticObjective}. For any approximate program defined by
  processor $i$ at time $t$, there exists a constant $\bar{\epsilon}_{it} > 0$
  such that $z^{[i]}(t)$ is the unique maximizer of the family of strictly
  concave objective functions $ J_{\epsilon}(z) := c^{T}z -
  \frac{\epsilon}{2}\|z \|^{2}, \quad \epsilon \in [0,\bar{\epsilon}_{it} ], $
  over the set of constraints $B^{[i]}(t) \cup Y^{[i]}(t)$.
  One can now always find $\underline{\epsilon} > 0$ such that
  $\underline{\epsilon} \leq \bar{\epsilon}_{it}$ for all $i\in \{1,\ldots,n\}$
  and $t\geq 0$.
  To see claim (iv), note that adding cutting-planes, either by receiving them
  from neighbors (S2) or by generating them with the oracle (S3), can only
  decrease the value of the strictly concave objective function $\Je(\cdot)$ and
  the basis computation in (S4) keeps, by its definition, the value of
  $\Je(\cdot)$ constant.
  Finally, (v) can be seen as follows. Starting at any time $t$ at some
  processor $i$, at time $t+1$ all processors in $l \in \innbrs(i,t)$ received
  the basis of processor $i$, and compute a query point that satisfies
  $\Je(z^{[l]}(t+1)) \leq \Je(z^{[i]}(t))$ for all $l \in \innbrs(i,t)$. This
  argument can be repeatedly applied to see that, in the static, strongly
  connected communication graph $\mc{G}_{c}$, at least after $\diam(\mc{G}_{c})$
  iterations, all processors in the network have an objective value smaller than
  $\Je(z^{[i]}(t))$. 
\end{proof}

\vspace{1em}

Next we present the proof of Lemma \ref{prop.Convergence}, Lemma \ref{prop.Agreement}, and Theorem \ref{thm.AsymptoticConvergence}.
The following proofs use the parameterized cost function $\Je(\cdot)$. However, for the clarity of presentation we will simplify our notation in the following proofs and write simply $J(\cdot)$ instead of $\Je(\cdot)$.

\subsubsection{Proof of Lemma \ref{prop.Convergence}}   

  All $z^{[i]}(t)$ are computed as maximizers of the common strictly concave
  objective function $J(\cdot)$ (Lemma \ref{prop.TechIssues} (iii)) and
  $J(\cdot)$ is monotonically non-increasing over the sequence of query points
  computed by a processor (Lemma \ref{prop.TechIssues} (iv)).
  Any sequence $\{ J(z^{[i]}(t)) \}_{t \geq 0}$, $i \in \{1,\ldots,n\}$, has
  therefore a limit point, i.e., $\lim_{t \rightarrow \infty} J(z^{[i]}(t))
  \rightarrow \bar{J}^{[i]}$.
  Since the sequence is convergent, it holds that $ \lim_{t\rightarrow \infty}
  \left(J(z^{[i]}(t)) - J(z^{[i]}(t+1)) \right) \rightarrow 0.$ By strict
  concavity of $J(\cdot)$ follows that $J(z^{[i]}(t)) - J(z^{[i]}(t+1)) >
  \sigma \|z^{[i]}(t) - z^{[i]}(t+1)\|^{2}_{2} $ for some $\sigma > 0$.
  Consequently, $\lim_{t\rightarrow \infty} \|z^{[i]}(t) - z^{[i]}(t+1)\|_{2}
  \rightarrow 0$ and the sequence of query points has a limit point, i.e., $
  \lim_{t \rightarrow \infty} \|z^{[i]}(t)- \bar{z}^{[i]}\|_{2} \rightarrow 0$.
  Suppose now, to get a contradiction, that $\bar{z}^{[i]} \notin \mc{Z}_{i}$.
  Then there exists $\delta > 0$ such that all $z$ satisfying $\|z -
  \bar{z}^{[i]}\|_{2} < \delta$ are not contained in $\mc{Z}_{i}$.
  Since $\lim_{t\rightarrow \infty} \|z^{[i]}(t) - \bar{z}^{[i]}\|_{2}
  \rightarrow 0$, there exists a time instant $T_{\delta}$ such that
  $\|z^{[i]}(t) - \bar{z}^{[i]}\|_{2} < \delta$ for all $t \geq T_{\delta}$, and
  thus $z^{[i]}(t) \notin \mc{Z}_{i}$ for $t \geq T_{\delta}$.
  But now, for all $t \geq T_{\delta}$ the oracle $\orc(z^{[i]}(t),\mc{Z}_{i})$
  will generate a cutting-plane according to \eqref{eqn.CuttingPlaneBasic},
  cutting off $z^{[i]}(t)$.
  According to \eqref{eqn.CuttingPlaneBasic}, it must hold that
  $a^{T}(z^{[i]}(t))z^{[i]}(t)-b(z^{[i]}) = s(z^{[i]}(t)) > 0$ and
  $a^{T}(z^{[i]}(t))z^{[i]}(t+1)-b(z^{[i]}) \leq 0$.
  This implies that $a^{T}(z^{[i]}(t))\left( z^{[i]}(t) - z^{[i]}(t+1) \right)
  \geq s(z^{[i]}(t))$ and consequently
  $ \| z^{[i]}(t) - z^{[i]}(t+1)\|_{2} \geq (\| a(z^{[i]}(t))\|_{2})^{-1}
  s(z^{[i]}(t)).  $
  By Assumption \ref{ass.Separator} (i) holds $\| a(z^{[i]}(t))\|_{2} < \infty$
  and thus $\lim_{t\rightarrow \infty} s(z^{[i]}(t)) \rightarrow 0$.
  As a consequence of Ass. \ref{ass.Separator} (ii) follows directly that
  $\bar{z}^{[i]} \in \mc{Z}_{i}$, providing the contradiction.  $ \hfill \blacksquare$ \\

\subsubsection{Proof of Lemma \ref{prop.Agreement}:} 
%
  Let $\bar{J}^{[i]}:= J(\bar{z}^{[i]})$ be the objective value of the limit
  point $\bar{z}^{[i]}$ of the sequence $\{z^{[i]}(t)\}_{t\geq 0}$ computed by
  processor $i$. We show first that the limiting objective values
  $\bar{J}^{[i]}$ are identical for all processors.
  Suppose by contradiction that there exist two processors, say $i$ and $j$,
  such that $ \bar{J}^{[i]} < \bar{J}^{[j]}. $ Pick now $\delta_{0} > 0$
  such that $\bar{J}^{[j]} - \bar{J}^{[i]} > \delta_{0}. $
  The sequences $\{J(z^{[i]}(t)) \}_{t \geq 0}$ and $\{J(z^{[j]}(t)) \}_{t
    \geq 0}$ are monotonically increasing and convergent. Thus, for every
  $\delta > 0$ there exists a time $T_{\delta}$ such that for all $t \geq
  T_{\delta}$, $ J(z^{[i]}(t)) - \bar{J}^{[i]} \leq \delta$ and $
  J(z^{[j]}(t)) - \bar{J}^{[j]} \leq \delta.$
  This implies that there exists $T_{\delta_{0}}$ such that for all $t \geq
  T_{\delta_{0}}$,
$$ J(z^{[i]}(t)) \leq \delta_{0} + \bar{J}^{[i]} < \bar{J}^{[j]}.  $$
Additionally, since the objective functions are non-increasing, it follows that
for any time instant $t' \geq 0$, $J(z^{[j]}(t')) \geq \bar{J}^{[j]}$. Thus,
for all $t \geq T_{\delta_{0}}$ and all $t' \geq 0$,
\begin{align} \label{eqn.Agreement.Cond1} J(z^{[i]}(t)) < J(z^{[j]}(t')).
\end{align}

Pick now $t_{0} \geq T_{\delta_{0}}$. For all $\tau \geq 0$ define now an index
set $I_{\tau}$ as follows: Set $I_{0} = \{i\}$ and for any $\tau \geq 0$ define
$I_{\tau}$ by adding to $I_{\tau-1}$ all indices $k$ for which there exist some
$l \in I_{\tau-1}$ such that $(k,l) \in E(t_{0}+\tau)$. Since, by assumption
$\mc{G}_{c}^{\infty}(t_{0})$ is strongly connected, the set $I_{\tau}$ will
eventually include all indices $1,\ldots,n$, and in particular there is
$\tau^{*}$ such that $j \in I_{\tau^{*}}$.
The algorithm is such that for all $l \in I_{\tau}$, $ J(z^{[l]}(t_{0}+\tau))
\leq J(z^{[i]}(t_{0})) $ and thus
\begin{align} \label{eqn.Agreement.Cond2} 
J(z^{[j]}(t_{0}+\tau^{*})) \leq
  J(z^{[i]}(t_{0})).
\end{align}
But \eqref{eqn.Agreement.Cond2} contradicts \eqref{eqn.Agreement.Cond1}, proving
that $ \bar{J}^{[i]} = \bar{J}^{[2]} = \cdots = \bar{J}^{[n]} =:
\bar{J}.$ Thus, it must hold that for all $i,j \in \{1,\ldots,n\}$,
$\lim_{t\rightarrow \infty} | J(z^{[i]}(t)) - J(z^{j}(t))| \rightarrow 0$.
From the strict concavity of $J(\cdot)$ follows that $|J(z^{[i]}(t)) -
J(z^{[j]}(t))| > \sigma \|z^{[i]}(t) - z^{[j]}(t)\|_{2}^{2},$ for some $\sigma
> 0$.  Therefore, $\lim_{t\rightarrow \infty }\|z^{[i]}(t) - z^{[j]}(t)\|_{2}
\rightarrow \infty$, which proves the theorem.  $ \hfill \blacksquare$ \\

\subsubsection{Proof of Theorem \ref{thm.AsymptoticConvergence}} 
%
  It follows from Lemma \ref{prop.Agreement} that the query points of all
  processors converge to the same query point, i.e., $\bar{z}^{[i]} = \bar{z}$
  for all processors $i$.  Now, we can conclude from Lemma
  \ref{prop.Convergence} that $\bar{z} \in \mc{Z}_{i}$ for all $i$ and thus
  $\bar{z} \in \mc{Z}$. It follows now from Lemma \ref{prop.TechIssues}, part
  (ii), that $\bar{z}$ is an optimal solution to \eqref{prob.Basic}.
  It remains to show that $\bar{z}$ is the optimal solution with minimal
  2-norm. Let $z^{*}$ be the optimal solution with minimal 2-norm. Then there
  exists an $\epsilon > 0$ such that the parameterized objective function satisfies $\Jeps(z^{*}) > \Jeps(z)$ for all $z \in
  \mc{Z}$ and $\Jeps(z^{[i]}(t)) \geq \Jeps(z^{*})$ for all $t$. With the same
  argumentation used for Lemma \ref{prop.TechIssues}, part (ii), we conclude
  that $\bar{z}$ is the unique solution maximizing $\Jeps(\cdot)$ over $\mc{Z}$,
  i.e., $\bar{z}$ is the optimal solution to \eqref{prob.Basic} with minimal
  2-norm.  $\hfill \blacksquare$ \\

\bibliographystyle{IEEEtran}   
\bibliography{Decomposition_long}

\end{document}